\documentclass[12pt]{article}
\usepackage{amsfonts}
\usepackage[utf8]{inputenc}
\usepackage{amssymb}
\usepackage{amsthm}
\usepackage{amsmath}
\usepackage{fancyhdr}
\usepackage{graphicx}
\usepackage{natbib}
\usepackage{caption}
\usepackage{geometry}
\usepackage{setspace}
\usepackage{subcaption}
\usepackage{multirow}
\usepackage{bbold}
\usepackage{setspace}

\newtheorem{theorem}{Theorem}[section]
\newtheorem{lemma}{Lemma}[section]
\newtheorem{assumption}{Assumption}[section]
\newtheorem{definition}{Definition}[section]
\DeclareMathOperator{\plim}{plim}
\DeclareMathOperator*{\argmax}{arg\,max} 

\begin{document}

\title{Using spatial modeling to address covariate measurement error}
\author{Susanne M. Schennach\thanks{%
smschenn@brown.edu, Department Box B, Dept. of Economics, Brown University,
Providence RI 02912.}\ \thanks{%
Support from NSF grant SES-1950969 is gratefully acknowledged. The authors
would like to thank Florian Gunsilius and Stelios Michalopoulos for useful
comments.}\ \ and Vincent Starck\thanks{%
V.Starck@lmu.de, Institut für Statistik LMU München, Ludwigstr. 33, 80539 München, Germany.}\ \thanks{Financial support from the European Research Council (Starting Grant No. 852332) is gratefully acknowledged.}}
\maketitle

\begin{abstract}
We propose a new estimation methodology to address the presence of covariate
measurement error by exploiting the availability of spatial data. The
approach uses neighboring observations as repeated measurements, after
suitably controlling for the random distance between the observations in a
way that allows the use of operator diagonalization methods to establish
identification. The method is applicable to general nonlinear models with
potentially nonclassical errors and does not rely on a priori distributional
assumptions regarding any of the variables. The method's implementation
combines a sieve semiparametric maximum likelihood with a first-step kernel
estimator and simulation methods. The method's
effectiveness is illustrated through both controlled simulations and an
application to the assessment of the effect of pre-colonial political
structure on current economic development in Africa.


\end{abstract}

\thispagestyle{fancy}

\newpage

\section{Introduction}

With the increasing availability of Graphical Information System (GIS) data %
\citep{zhou2017spatial} and network data \citep{paula:econet}, spatial
econometrics \citep{pinkse2010future, redding2017quantitative} is becoming
an increasingly influential field. Further, spatial setups readily
generalize to more abstract spaces, with the spatial dimensions representing
individual or product characteristics, and the increasing availability of
rich datasets with suitable covariates enables this avenue of research.

This paper identifies another advantage provided by the use of spatial
datasets. The inherent redundancy provided by numerous nearby observations
in spatial frameworks generates information that can be used to correct for
covariate measurement error and achieve consistency without requiring
additional information such as validation data or the knowledge of the
measurement error distribution. The method is very generally applicable, as
it allows for nonlinear models as well as non-classical measurement error %
\citep{schennach:annrev}. This is made possible in part by leveraging
identification results from \citet{hu2008instrumental} and %
\citet{hu2008identification} and in part by devising a scheme to generate
``virtual'' observations that can act as repeated measurements, from the
information provided by the observed sample.

Our approach is to be contrasted to others developed within the Kriging
literature (\citet{krige:thesis}, \citet{chiles:kriging}). Kriging is a
common method to carry out inference regarding spatial quantities in between
available measurements. While this approach has been extended to allow for
measurement error (e.g., \citet{cressie:kriking}), most of this line of
research does not consider the implications of using the mismeasured data as
a covariate. Methods that do consider covariates tend to rely on
distributional assumptions and linearity (e.g., \citet{szipiro:misalligned})
or achieve bias reduction but not consistency (e.g., %
\citet{carroll:spatialsimex}).

While the approach we take is reminiscent of using lags or leads as repeated
measurements in the context of time series or panel data econometrics %
\citep{hu2008nonparametric, cunha2010estimating, griliches:panel}, a
corresponding approach in a spatial framework is not currently available,
due to significant conceptual and algorithmic challenges. Unless the spatial
data happens to lie on a fixed grid (a rare occurrence), there is no spatial
analog of a fixed time-shift, since the spacing between data points is a
random quantity.\footnote{%
Although there is long tradition of using neighboring observations as
instruments in the spatial or network literature (e.g., %
\citet{prusha:spatialiv}, \citet{bramoulle:peernet}), it is well-known that
instruments cannot be used to correct for measurement error in general
nonlinear models \citep{Amem:IV}. Furthermore, such instruments cannot
simply be converted into suitable repeated measurements, because the
variable distance between the observations causes an unknown bias in the
measurement error that is difficult to account for.} This randomness
generally invalidates the use of neighboring observations as proper repeated
measurements.

We propose to overcome this challenge by expressing the joint density of the
dependent variable, the mismeasured variable, and its value at a neighboring
point, as a function of the distance to the neighboring point. This approach
enables us to condition on a fixed distance to generate a virtual repeated
measurement with statistical properties suitable to play the role of the
counterpart of a fixed lag repeated measurement. We show that any given
fixed distance permits the identification of the model, but efficiency
considerations suggest the use of a weighted average of estimators coming
from different distances. The effectiveness and feasibility of this approach
is demonstrated through a controlled simulations study.

The estimator is applied to provide further corroboration to an important
study \citep{michalopoulos2013pre} seeking to quantify whether pre-existing
political structures of ethnic groups in the pre-colonial Africa still have
a significant impact on contemporary economic development. The main
descriptor of the political structure is a measure of centralization of
political power (i.e., whether decisions are made at a very local level in a
decentralized fashion or at a broader level in a centralized fashion).

The conclusions of this study, however, rest on the accuracy of such
estimated centralization measures. Our approach specifically enables us to
quantify the relevant error distributions and obtain measurement
error-robust estimates by exploiting the spatial nature of the data to
construct repeated measurements of centralization using data points in the
geographical vicinity of each observation. Remarkably, our results reinforce
those of the authors by uncovering an even stronger relationship between
pre-colonial centralization and contemporary development. This points to a
significant potential for our method to circumvent measurement error issues
in a broader range of similar applications.

The paper is organized as follows. Section 2 describes the setup, its
motivation and establishes identification, section 3 discusses the estimator
and its implementation, section 4 provides simulations to assess the
performance of the estimator, section 5 applies our estimator to the study
of political complexity on current economic development, and section 6
concludes.

\section{Setup and Identification}

Throughout the text, we denote random variables (or random functions) by
upper case letters, while the corresponding lower case letter denotes
specific values. We also denote (conditional) densities by $f$ with suitable
random variable subscripts and assume their existence, relative to a
suitable dominating measure.

We consider a spatial setup, denoting (potentially abstract\footnote{%
Abstract location examples could include product or individual
characteristics. In \textquotedblleft big data\textquotedblright\ settings,
low-dimensional abstract location variables could be extracted from
high-dimensional covariates through linear \citep{jolliffe:pcamarket} or
nonlinear \citep{schennach:nlfact, gunsilius2023independent} component analysis.}) locations
by $S$, taking values in some set $\mathcal{S}\subset \mathbb{R}^{d_{s}}$.
The model of interest is 
\begin{equation}
Y(S)=g(X^{\ast }(S))+U(S) \label{eqYvsXs}
\end{equation}%
where $Y(S)$ is the dependent variable, $X^{\ast }(S)$ is an unobserved
(potentially multivariate) regressor, $U(S)$ is the model error. We observe a sample $(S_{i},X_{i}%
\mathrel{\overset{\makebox[0pt]{\mbox{\normalfont\tiny\sffamily def}}}{=}}%
X(S_{i}),Y_{i}\mathrel{\overset{\makebox[0pt]{\mbox{\normalfont\tiny%
\sffamily def}}}{=}}Y(S_{i}),i=1,...,n)$ where $X(S)$ is an
error-contaminated version of $X^{\ast }(S)$: 
\begin{equation}
X(S)=X^{\ast }(S)+V(S).  \label{generateX}
\end{equation}%
Although, for simplicity, we do not make this explicit in this section,
covariates could be included in our identification analysis
by making all assumptions and densities
conditional on the covariates. We shall re-introduce an explicit dependence on the covariates when considering asymptotic properties.

We are interested in the conditional distribution $f_{Y(s)|X^{\ast
}(s)}(y|x^{\ast })$, which will allow us to recover the function $g$. Since $%
X^{\ast }(s)$ is unobserved due to measurement error, this density is not
directly revealed by the data and its identification will be secured through
availability of repeated measurements. Here we observe that spatial
processes provide natural candidates for repeated measurements for $X(s)$
through neighboring observations $X(s+\Delta s)$, where $\Delta s$ is some
fixed vector-valued shift. Our identification argument relies on one
specific value of $\Delta s$, but, in fact, there are potentially an
infinite number of repeated measurements (for different $\Delta s$), which
can be used to improve efficiency.

In our approach, the disturbances satisfy the following:

\begin{assumption}[Exclusion restrictions]
\label{asscondindep} The random variables $Y(s), X(s), X(s+\Delta s)$ are
mutually independent conditional on $X^{*}(s)$ for any $s$ and any $\Delta s$
such that $\Vert \Delta s \Vert > \Delta s_{0}$ for some given known $\Delta s_{0} \geq 0$.
\end{assumption}

The assumption is reasonable if, for instance, the measurement error is due
to devices at different locations, each imperfectly measuring the covariate
(with $\Delta s_{0}=0$), or is due to some local noise whose spatial correlation decays
with distance quicker than that of the underlying process (in which case $\Delta s_{0}> 0$
represents the 'locality' of the error process). In practice, checking
the validity of the assumption hinges on an assessment of the nature of the
measurement error process. 
The fact that the assumption involves a spatial shift $\Delta s$ will allow
us to consider a neighboring observation as a repeated measurement. Note
that, while Assumption \ref{asscondindep} implicitly places restrictions on
the spatial dependence of the measurement error process $V(s)$, we place no
such restrictions on the generating processes of $U(s)$.

To precisely state our identification results, we first require some basic
regularity conditions about the distributions.

\begin{assumption}[Existence of bounded densities]
\label{densities} For a given $\Delta s$, the joint distribution of $Y(s)$
and $X(s)$, $X(s+\Delta s)$ and $X^{*}(s)$, admits a bounded density $%
f_{Y(s),X(s),X(s+\Delta s),X^{*}(s)}$ with respect to a dominating measure
of the form $\mu_Y \times \mu_X \times \mu_X \times \mu_X$ where $\mu_Y$ is
unrestricted while $\mu_X$ could be either the Lebesgue measure or a
discrete measure supported on a finite set of points. All marginal and
conditional densities are also bounded.
\end{assumption}

These conditions on the density allow us to cover both continuous and
discrete $X(s)$ (and $X^{\ast }(s)$), thus covering either measurement error
or misclassification. Although our presentation in the main text covers
these two cases within a common overarching notation, they demand
significantly different treatments both on a theoretical and implementation
level (see \citet{hu2008instrumental} and \citet{hu2008identification}, for
the continuous and discrete cases, respectively), which are reflected in our
formal proofs in Appendix \ref{appproof} and in our implementation procedure.
Few restrictions are placed on the nature of the distribution of $Y(s)$.

We also impose

\begin{assumption}[Centering]
\label{centering} For a known functional $M_x$, we have $M_x[f_{X(s)|X^{*}(s)}( \cdot | x^{*})] = x^{*}$ for any $x^{*}$.
\end{assumption}

This type of assumption is commonly made in the context of nonclassical
measurement error models \citep{hu2008instrumental} and extends standard
conditional mean assumptions to more general centering concepts (e.g. mode,
median or general quantiles). For conciseness, we state here a condition
that is sufficient to transparently cover both the discrete and continuous
cases, although it could be relaxed in the discrete case (see %
\citet{hu2008identification}).

We also require nonparametric analogues of rank conditions, which have a
long history in the nonparametric instrumental variable literature %
\citep{newey:npiv, hallhorowitz:npiv, hu2008instrumental}

\begin{assumption}[Injectivity of operators]
\label{injectivity} The operators $L_{X(s)|X^{*}(s)}$ and $L_{X(s+\Delta
s)|X^{*}(s)}$ are injective, where $L_{B|A}$ is defined through its action
on a function $h$ by $[L_{B|A}h] (b) \mathrel{\overset{\makebox[0pt]{\mbox{%
\normalfont\tiny\sffamily def}}}{=}} \int f_{B|A}(b|a) h(a) d\mu_X(a)$.
\end{assumption}

In the discrete case, this condition reduces to a familiar full rank
condition on the matrices of conditional probabilities $%
f_{X(s)|X^{*}(s)}(x|x^{*})$ and $f_{X(s+\Delta s)|X^{*}(s)}(x|x^{*})$
(indexed by $x$ and $x^{*}$).

For the outcome variable $Y(s)$, a weaker rank-like condition is sufficient:

\begin{assumption}[Outcome variation]
\label{variation} For all $x_{1}^{*} \neq x_{2}^{*}$, the set $\{y :
f_{Y(s)|X^{*}(s)}(y|x_{1}^{*}) \neq f_{Y(s)|X^{*}(s)}(y|x_{2}^{*})\}$ has
positive probability.
\end{assumption}

\citet{hu:finmix} observe that, in the discrete case, these conditions
provide easily verifiable conditions that reach Kruskal's minimum rank bounds for the identification of discrete
probability models defined in terms of three-way arrays \citep{kruskal:3way}%
. As noted in \citet{schennach:annrev}, in the continuous case, these two
conditions also reach a continuous analog of Kruskal's minimum rank bounds.

We are now ready to state our main identification result (proven in
Appendix \ref{appproof}):

\begin{theorem}[Identification]
\label{thmmain} Under assumptions \ref{asscondindep} to \ref{variation}, the
(conditional) densities $f_{Y(s)|X^{\ast }(s)}$, $f_{X(s)|X^{\ast }(s)}$, $%
f_{X(s+\Delta s)|X^{\ast }(s)}$, and $f_{X^{\ast }(s)}$ are identified
(almost everywhere) from the observed joint density $f_{Y(s),X(s),X(s+\Delta
s)}$.
\end{theorem}

From this result, any model (such as Equation (\ref{eqYvsXs})) that seeks to
determine a relation between $Y$ and $X^{\ast }$ is also identified. The
practical use of this identification result obviously requires the
determination of the density $f_{Y(s),X(s),X(s+\Delta s)}$. When locations
are regularly spaced, $\Delta s$ can be fixed so that knowledge of the
sample $(Y(S_{i}),X(S_{i}),X(S_{i}+\Delta s))$ is sufficient for estimation.
However, as noted earlier, if locations $S_{i}$ have random spacings, there
may not be pairs of observations exactly $\Delta s$ apart from each other.
In this case, we view the density of interest, $f_{Y(s),X(s),X(s+\Delta s)}(y,x,z)$, as a smooth function of $\Delta s$ that can be estimated via kernel smoothing, thanks to the identity:
\begin{equation}
f_{Y\left( s\right) ,X\left( s\right) ,X\left( s+\Delta s\right) }\left( y,x,z\right) =\lim_{h\longrightarrow 0}\int \frac{1}{h}K\left( \frac{u-\Delta s}{h}\right) f_{Y\left( s\right) ,X\left( s\right) ,X\left( s+u\right) }\left( y,x,z\right) du, \label{eqdenscondds}
\end{equation}%
under the
assumption that locations are drawn from some continuous density over space.
In some applications, isotropy can help reduce the dimensionality for
density estimation (in which case $\Vert\Delta s \Vert$ becomes the relevant
argument).
Naturally, this approach relies on a stationarity assumption for estimation:

\begin{assumption}[Stationarity]
The process $(Y(s), X(s))$ is strictly stationary.
%The joint distribution of $Y(s),X(s),X(s+\Delta s)$ does not depend on $s\in 
%\mathcal{S}$ for any given $\Delta s\in \Delta \mathcal{S}$.
\end{assumption}

% While stationarity assumptions have been criticized in spatial applications \citep{pinkse2010future} due to inherent geographic inhomogeneities, this assumption can be weakened by instead considering a conditional density
Although spatial stationarity assumptions have been criticized in
applications \citep{pinkse2010future} due to inherent geographic
inhomogeneities, they have frequently been invoked when establishing spatial
results such as in GMM estimation \citep{conley1999gmm}, central limit
theorems \citep{bolthausen1982central, lahiri2003central}, density
estimation \citep{carbon1997kernel, hallin2004kernel}, and more recently
functional-coefficient spatial autoregressive models %
\citep{sun2016functional, sun2018estimation}.

Furthermore, the stationarity requirement can be substantially weakened by
viewing the density of interest as a conditional density 
\begin{equation}
f_{Y(s), X(s), X(s+\Delta s)|T}(y,x,z|t)
\end{equation}
where $T$ is a position-dependent variable that controls for the source of
the lack of stationarity. All above assumptions and results are then
understood to be conditional on $T$ (which is suppressed in the notation,
for simplicity). For instance, $T$ could be the distance to the nearest body
of water, the degree of a node in graph/network applications or controls for
treatment status or law enactments.\footnote{%
It should be stated that high-dimensionality of $T$ may have an impact on
estimation accuracy, due to the data needs associated with high-dimensional
density estimations. In practice, dimensionality of $T$ may thus be limited
by the size of the available data.}

It is even possible, in principle, to fully relax stationarity by
partitioning the space of $S$ through a grid of resolution $b$ and letting $T $ denote which grid \textquotedblleft box\textquotedblright\ point $S$
belongs to. If we let $b\rightarrow 0$ as $n\rightarrow \infty $,
stationarity conditional on $T$ will hold asymptotically under suitable
regularity conditions regarding the generating process. It is also possible
to replace partitioning into boxes by suitable kernel smoothing. For either
approaches, the key variance-bias trade-off to achieve is to simultaneously
ensure that (i) the number of observations within a region of linear extent $%
b$ still goes to infinity as sample size grows and (ii) the changes in the
distribution due to dependence on $T$ becomes asymptotically negligible
within a region of linear extent $b$. These considerations will typically
require, respectively, that $nb\longrightarrow \infty $ and that $b=o\left(
n^{-1/4}\right) $ along with twice differentiability of the dependence of
the density of all variables on $T$. We however leave a formal analysis of
these extensions, along with all necessary regularity conditions, for future
work to avoid obscuring the main ideas.

\section{Estimator and Implementation}

Estimation is based on the identity 
\begin{eqnarray}
& & f_{Y(s),X(s),X(s+\Delta s)}(y,x,z)  \label{eqfactor} \\
&=& \int f_{Y(s)|X^{*}(s)}(y|x^{*}) f_{X^{*}(s)}(x^{*})
f_{X(s)|X^{*}(s)}(x|x^{*}) f_{X(s+\Delta s)|X^{*}(s)}(z|x^{*}, \Delta s) d
\mu_X(x^{*}),  \notag
\end{eqnarray}
implied by conditional independence (Assumption \ref{asscondindep}). Theorem %
\ref{thmmain} implies that this integral equation, for a given left-hand
side density, has a unique solution. Hence, we can use the right-hand side
of (\ref{eqfactor}) to construct an estimator analogous to a maximum
likelihood estimator (MLE) in terms of 4 unknown densities to be estimated.
In the misclassification case ($\mu_X$ discrete), the densities $%
f_{X(s)|X^{*}(s)}(x|x^{*})$, $f_{X(s+\Delta s)|X^{*}(s)}(z|x^{*})$ and $%
f_{X^{*}(s)}(x^{*})$ can be parametrized as a matrix (or a vector) of
probabilities, as in \citet{hu2008identification}. In the continuous $\mu_X$
case, the densities are represented by a sieve approximation, as in %
\citet{hu2008instrumental}. 
%In this section, we focus on challenges pertaining to the latter, continuous, case.

One important aspect of our approach that is distinct from earlier work
(such as \citet{hu2008instrumental}) is the fact that $X(s+\Delta s)$ is not
a repeated measurement in the usual sense, because we only have access to
its estimated density, not its specific value at each sample point. We
address this by sampling pseudo-observations from the density 
\begin{equation*}
f_{X(s+\Delta s)|Y(s), X(s)}(z | y,x)= \frac{f_{Y(s),
X(s), X(s+\Delta s)}(y,x,z)} {\int f_{Y(s), X(s),
X(s+\Delta s)}(y,x,z) d\mu_X(z)}
\end{equation*}
where the right-hand side can be estimated from kernel smoothing, as in
Equation (\ref{eqdenscondds}), for some pre-specified $\Delta s$. For
estimation purposes, our sample then consists of $Y_i\mathrel{\overset{%
\makebox[0pt]{\mbox{\normalfont\tiny\sffamily def}}}{=}} Y(S_i)$, $X_i%
\mathrel{\overset{\makebox[0pt]{\mbox{\normalfont\tiny\sffamily def}}}{=}}
X(S_i)$ and $Z_i$ drawn from an estimate of $f_{X(s+\Delta s)|Y(s),
X(s)}(z | Y_i,X_i)$ for $i=1,\ldots, n$. One could, of
course, draw multiple pseudo-observations per data point to reduce the
simulation noise, although we did not find this to be necessary in our
application and simulations study. In cases where the data is very dense
along the spatial dimension, it may be possible to directly draw at random
from neighboring observations that lie within an asymptotically vanishing
tolerance $h$ of a given shift $\Delta s$ instead of first estimating a
conditional density. This scheme, however, makes it impossible to exploit
the faster convergence enabled by using higher-order kernels and the
noise-reduction arising from averaging over similar observation pairs at
different locations.

We then use a semiparametric sieve maximum likelihood estimator (MLE)%
\citep{Shen:sieve} of the form: 
\begin{equation}  \label{maximization}
(\hat{\theta}, \hat{\eta}, \hat{f}_{1}, \hat{f}_{2}, \hat{f}_{3}) = \argmax%
_{(\theta, \eta, f_{1}, f_{2}, f_{3})} \sum_{i=1}^{n} \ln
L(Y_i,X_i,Z_i;\theta,\eta,f_1,f_2,f_3)
\end{equation}
where the maximization is performed under suitable constraints detailed
below and where 
\begin{equation}  \label{eqdefL}
L(y,x,z;\theta,\eta,f_1,f_2,f_3) \mathrel{\overset{\makebox[0pt]{\mbox{%
\normalfont\tiny\sffamily def}}}{=}} \int_{\mathcal{X}^{*}}
f_{y|x^*}(y|x^{*}; \theta, \eta) f_{1}(x^{*}) f_{2}(x|x^{*}) f_{3}(z|x^{*}) dx^{*}
\end{equation}
where $\mathcal{X}^{*}$ denotes the support of $X^{*}$. In (\ref{eqdefL}),
the density $f_{Y(s)|X^{*}(s)}(y|x^{*})$ is indexed by $\theta$, the
parameter of interest and $\eta$, some nuisance parameter. In our setup, $%
\theta$ could specify the shape of the function $g$ in Equation (\ref{eqYvsXs}),
while $\eta$ could specify the density of the disturbance $U(S)$ 
(other ways to separate $\theta$ and $\eta$ are possible: for instance, $%
\theta$ could represent an average derivative, while $\eta$ includes both
the density of $U(S)$ and degrees of freedom of $g$ which do not affect the
average derivative. See \citet{hu2008instrumental} for more details). No
such separation is imposed on the remaining densities $(f_1,f_2,f_3)$, which
are all considered nuisance parameters. Note that only $f_3$ depends on the
shift $\Delta s$. For conciseness, we shall suppress this dependence in the notation whenever it is clear from the context. The parameter of interest $\theta$ is considered
finite dimensional, while all other parameters are infinite dimensional and
approximated through sieves in finite samples. This setup reflects most
empirical studies and will enable the development of an asymptotic theory
for asymptotic normality and root-$n$ consistency (in the next section).

The optimization in Equation (\ref{maximization}) must be performed under
some constraints in order to enforce assumptions needed for identification
as well as basic properties of densities. To enforce nonnegativity
constraints, we actually model the square root of densities, so that their
respective squares are automatically positive: 
\begin{equation}
f_{1}^{\frac{1}{2}}(x^{*}) = \sum_{i=1}^{i_{n}+1} \alpha_{i} p_{i, 1}(x^{*})
= \boldsymbol{p}_{1}(x^{*})^{\prime}\boldsymbol{\alpha}
\end{equation}
\begin{equation}
f_{2}^{\frac{1}{2}}(x|x^{*}) = \sum_{i=1}^{i_{n}+1} \sum_{j=1}^{j_{n}+1}
\beta_{ij} p_{i, 2}(x-x^{*}) q_{j}(x^{*}) = \boldsymbol{p}%
_{2}(x-x^{*})^{\prime}\boldsymbol{\beta} \boldsymbol{q}(x^{*})
\end{equation}
\begin{equation}
f_{3}^{\frac{1}{2}}(z|x^{*}) = \sum_{i=1}^{i_{n}+1} \sum_{j=1}^{j_{n}+1}
\gamma_{ij} p_{i, 3}(z-x^{*}) q_{j}(x^{*}) = \boldsymbol{p}%
_{3}(z-x^{*})^{\prime}\boldsymbol{\gamma} \boldsymbol{q}(x^{*})
\end{equation}

Let $x^{\ast }\in \lbrack 0,l_{x}]$, $(x-x^{\ast })\in \lbrack -l_{1},l_{1}]$%
, and $(z-x^{\ast })\in \lbrack -l_{2},l_{2}]$. With Fourier series, we have 
\newline
$p_{k,1}(a)=\ \cos (\frac{k2\pi }{l_{x}}a)$ or $p_{k,1}(a)=\ \sin (\frac{%
k2\pi }{l_{x}}a)\ \forall k>1$ for the univariate density while \newline
$p_{k,m}(a)=\ \cos (\frac{k\pi }{l_{m}}a)$ or $p_{k,m}(a)=\ \sin (\frac{k\pi 
}{l_{m}}a)\ \forall k>1$ and $m$ $\in \{1,2\}$, and $q_{k}(a)=\ \cos (\frac{%
k\pi }{l_{x}}a)$. $f(y_{i}|x^{\ast })$ can be specified similarly or be
fully parametric. In the following, we use both cosines and sines in numbers 
$\frac{i_{n}}{2}$ each (the first of the $(i_{n}+1)$ terms being the
constant). The number of terms in each series expansion can be determined
using existing data-driven methods (e.g. \citet{schennach:berkson}, %
\citet{laan:mlecv}).

Since non-negativity constraints are automatically satisfied by squaring, $%
M_{x}[f_{2}(\cdot |x^{\ast })]=x^{\ast }$ and densities integrating to 1 remain to enforce. We
proceed as follows. Consider for instance the constraint that the density $%
f_{2}=(\sum_{i=1}^{i_{n}+1}\sum_{j=1}^{j_{n}+1}p_{i}\Lambda _{ij}q_{j})^{2}$
(dropping the potential RHS subscript and arguments to ease notation)
integrates to 1. In matrix form, this reads $f_{2}=(p^{\prime }\Lambda
q)^{2}=q^{^{\prime }}\Lambda ^{^{\prime }}pp^{^{\prime }}\Lambda q$ so that,
if one uses an orthonormal basis, 
\begin{equation}
\int f_{2}(x|x^{\ast })\ dx=q^{^{\prime }}\Lambda ^{^{\prime }}I\Lambda
q=(q^{\prime }\otimes q^{\prime })vec(\Lambda ^{^{\prime }}\Lambda )
\end{equation}%
For the vector of orthogonal functions $B(x^{\ast })=[1\ \cos (x^{\ast })\
...\ \cos (2j_{n}x^{\ast })]^{\prime }$ and the transformation $T$ that
satisfies $TB(x^{\ast })=q(x^{\ast })\otimes q(x^{\ast })$, we obtain the
restriction $B^{\prime }(x^{\ast })T^{\prime }vec(\Lambda ^{^{\prime
}}\Lambda )=1$, i.e. $[T^{\prime }vec(\Lambda ^{^{\prime }}\Lambda )]_{11}=1$
and $[T^{\prime }vec(\Lambda ^{^{\prime }}\Lambda )]_{k1}=0$ for $k>1$.
Other constraints can be treated similarly; after a bit of algebra, it is
simple to implement the constraint brought by the functional, whether the
expected value, the median, the mode, or a percentile.

Solving the optimization problem (\ref{maximization}) subject to the
constraints delivers $\hat{\theta}_{\Delta s}$ for the chosen $\Delta s$.
Although any single nonzero value of $\Delta s$ delivers a consistent
estimator, its efficiency can be improved by combining the information
provided by all other distances. Since kernel estimates at two nearby points
are asymptotically uncorrelated, an asymptotically optimal linear
combination of the different $\hat{\theta}_{\Delta s}$ simply involves
weights inversely proportional to the variance of the corresponding
estimators. This approach is supported by our simulation experiments in
finite sample, which reveal only weak correlation between the estimation
errors of estimators based on different distances. Naturally, to ensure that
this asymptotic behavior is reached, it is recommended that the spacing
between the different $\Delta s$ be selected so that it converges to zero
slower than the bandwidth does, as sample size grows.

While it is beyond the scope of this work to provide a fully formally
justified data-driven $\Delta s$\ selection method, we can nevertheless
provide some guidance to practitioners. For simplicity, consider the case
where suitable $\Delta s$\ are selected purely based on a criterion of the
form $\left\Vert \Delta s\right\Vert \in \left[ \Delta s_{\min },\Delta
s_{\max }\right] $.

Start with $\Delta s_{\max }=\Delta s_{\min }=2h$, where 
$h$\ denotes the bandwidth used to smooth along $\Delta s$.\ Then gradually
increase $\Delta s_{\max }$ (while keeping $\Delta s_{\min }$\ fixed) in
steps of $2h$ and monitor the decrease in the estimated standard error on $%
\hat{\theta}$. The steps of length $2h$\ are motivated by the fact that
point estimates obtained with $\Delta s$\ that differ by that amount should
be roughly uncorrelated. Continue until a further increase in $\Delta s_{\max
} $\ either increases the standard error or only provides negligible
improvement. Next, increase $\Delta s_{\min }$\ (while keeping $\Delta
s_{\max }$\ fixed) in steps of $2h$ until finding two consecutive values of
$\Delta s_{\min }^{1},\Delta s_{\min }^{2}$\ yielding corresponding point
estimates $\hat{\theta}_{\Delta s_{\min }^{1}},\hat{\theta}_{\Delta s_{\min
}^{2}}$ that differ by less than a given small fraction $\phi $ of the
estimated standard error $\hat{\sigma}_{\hat{\theta}_{\Delta s_{\min }^{2}}}$
on $\hat{\theta}_{\Delta s_{\min }^{2}}$\ and report $\hat{\theta}_{\Delta
s_{\min }^{2}}$. If no such $\Delta s_{\min }^2$ is found, then iterate the procedure described above, now with $\Delta s_{\min }$ set to the largest $\Delta s_{\min }^2$ considered so far and again increase $\Delta s_{\max }$, etc.

The rationale for this approach is that the adjustment of $%
\Delta s_{\max }$ seeks to optimize the variance, while the adjustment in $%
\Delta s_{\min }$\ aims to control the bias. In some sense, the optimal
choice of $\Delta s_{\max }$\ is not that critical, as only efficiency could
suffer. The choice of $\Delta s_{\min }$\ is more important and enforces the
practitioner's tolerance for bias via the threshold $\phi $, which should be
sufficiently small so that statistically significant findings would not be
overturned if the bias were indeed this large.

\section{Inference}

Our estimator's hybrid nature (i.e. with $Z_{i}$ drawn from a kernel density
estimator fed into a sieve semiparametric MLE) makes its asymptotic analysis
much more involved than an application of standard results on sieve MLE and
complicates an explicit calculation of its asymptotic variance. To address
this, we establish that the construction of the $Z_{i}$ still yields an
estimator that admits an asymptotically linear representation, provided that
the corresponding (infeasible) sieve estimator with observed $Z_{i}$ has
that property. This result, stated formally below, will simultaneously
ensure asymptotic normality, root $n$ consistency, and asymptotic validity
of the bootstrap for our estimator.

To state our main asymptotic result, we define a profiled likelihood that
focuses on the parameter $\theta $ of interest:%
\begin{equation}
\mathcal{L}\left( \theta ,f\right) =E\left[ \ln L\left( Y,X,Z,W;\theta
,\omega \left( \theta \right) \right) \right]   \label{eqprofL}
\end{equation}%
for%
\begin{equation*}
\omega \left( \theta \right) =\arg \max_{\omega \in \Omega }E\left[ \ln
L\left( Y,X,Z,W;\theta ,\omega \right) \right] 
\end{equation*}%
with $Z$ distributed according to the conditional density $f_{X(s+\Delta
s)|Y(s),X(s),W}(z|y,x,w)$ of the repeated measurement,
which we denote by $f_{Z|X,Y,W}$ for simplicity (or simply $f$ when the
context avoids any confusion). The parameter $\omega \equiv \left( \eta
,f_{1},f_{2},f_{3}\right) $ denotes all the nuisance parameters, which
belong to some set $\Omega $ imposing suitable regularity conditions. Let $%
\theta _{0}$ and $f_{0}$ denote the true values of $\theta $ and $f$,
respectively. We explicitly include the possible dependence of the
likelihood function and the density $f_{Z|X,Y,W}$ on a vector of observed
covariates $W$. These can be incorporated into the definition of our
likelihood (Equation (\ref{eqdefL})) by conditioning all densities on the
covariates.

The empirical counterpart of (\ref{eqprofL}) is:%
\begin{equation}
\mathcal{\hat{L}}\left( \theta ,\hat{f}\right) =\frac{1}{n}\sum_{i=1}^{n}\ln
L\left( Y_{i},X_{i},Z_{i},W_{i};\theta ,\hat{\omega}\left( \theta \right)
\right) 
\end{equation}%
for%
\begin{equation*}
\hat{\omega}\left( \theta \right) =\arg \max_{\hat{\omega}\in \Omega _{n}}%
\frac{1}{n}\sum_{i=1}^{n}\ln L\left( Y_{i},X_{i},Z_{i},W_{i};\theta ,\hat{%
\omega}\right) 
\end{equation*}%
with $Z_{i}$ drawn from the estimated density $\hat{f}\equiv \hat{f}%
_{Z|X,Y,W}$ and the maximum is taken over a sample-size dependent sieve
space $\Omega _{n}$ (as described in the previous section). We define $\hat{%
\theta}=\arg \max_{\theta }\mathcal{\hat{L}}\left( \theta ,\hat{f}\right) $,
for some estimated $\hat{f}$.

To accommodate possible covariates, we allow $\hat{f}$ to depend on an
estimated finite dimensional nuisance parameter $\hat{\kappa}$ (whose true
value is denoted $\kappa _{0}$). We let $\kappa $ be finite dimensional to
reflect the fact that most empirical researchers would want to include
covariates through a parametric model to mitigate a possible curse of
dimensionality. In the same spirit, the dependence on $W$ and $\kappa $ is
assumed to have an index structure (where it is understood that the
definition of $Z$ is $\Delta s$-dependent):

\begin{assumption}
\label{assindex}The variables $Y,X,Z$, are generated through:%
\begin{equation}
(Y,X,Z)=G\left( (\tilde{Y},\tilde{X},\tilde{Z}) ,W,\kappa \right),   \label{eqgenV}
\end{equation}%
where $G$ is a known link function (depending on an unknown parameter 
$\kappa $) that is one-to-one in its first argument and the $(\tilde{Y},\tilde{X},\tilde{Z})$ are
jointly drawn from an (unknown) density $f_{\tilde{Y}\tilde{X}\tilde{Z}%
}\left( y,x,z\right) $ and independent from $W$.
\end{assumption}

This assumption effectively breaks our model down into a nonparametric component, involving the main variable of interest ($Y,X,Z$) and a parametric component, involving the covariates $W$. This general form is not very restrictive, since the researcher is allowed to specify the dependence on $W$ as flexibly as demanded by the problem at hand, while keeping in mind the usual bias-variance trade-off. This approach aims to provide a practical way to keep under control a possible curse of dimensionality in the presence of many covariates. In Section \ref{secapp}, we provide a specific example of link function.
Readers wishing to consider a covariate-free version of our estimator can simply ignore Assumption \ref{assindex} along with any $W$- and $\kappa$-dependence and take the function $G$ to be identity function in our treatment below.
%^ SMS: added

This structure suggests the following estimation procedure: First, letting $G^{-1}$ denote inverse with respect to the first vector-valued argument $(Y,X,Z)$, we define, for a given spacing $\Delta s$,
\begin{eqnarray}
&&\hat{f}_{\tilde{Y},\tilde{X},\tilde{Z}}\left( \tilde{y},\tilde{x},%
\tilde{z};\Delta s\right) \label{eqfhattilde} \\
&&=n^{-1}h^{-2d_{x}-d_{y}-d_{s}}  
\sum_{i=1}^{n}K_{yxz,s}(h^{-1}\left(
G^{-1}\left( (Y_{i},X_{i},Z_{i}),W_{i},\hat{\kappa}\right) -(\tilde{y},\tilde{x},\tilde{z})\right), h^{-1}\left( \Delta S_{i}-\Delta s\right) ).  \notag
\end{eqnarray}%
for some kernel function $K_{yxz,s}$ and an estimated $\hat{\kappa}$ and where $\Delta S_{i}$ denote spacings observed in the sample.
Typically, $\hat{\kappa}$ is obtained by regressions of simple functions of $X,Y,Z$ onto $W$.\ Next, let%
\begin{equation}
\hat{f}_{Y,X,Z|W}\left( y,x,z|w;\Delta s\right) =\hat{f}_{\tilde{Y},%
\tilde{X},\tilde{Z}}\left(
G^{-1}\left( (y,x,z) ,w,\hat{\kappa}\right); \Delta s\right) \hat{J}\left( y,x,z,w,\hat{\kappa}\right) ,
\label{eqfhatnotilde}
\end{equation}%
where $\hat{J}\left( y,x,z,w,\hat{\kappa}\right) =\left( \det \nabla_{(y,x,z)}^{\prime }G\left( (y,x,z),w,\hat{\kappa}\right) \right) ^{-1}$\ is a Jacobian term. We can similarly
construct an estimator $\hat{f}_{Y,X|W}\left( y,x|w;\Delta s\right) 
$ (with a kernel function $K_{yx}$) and set%
\begin{equation*}
\hat{f}_{Z|Y,X,W}\left( z|y,x,w;\Delta s\right) =\frac{\hat{f}%
_{Y,X,Z|W}\left( y,x,z|w;\Delta s\right) }{\hat{f}_{Y,X|W}\left( y,x|w;\Delta s\right) },
\end{equation*}%
from which the repeated measurements $Z_{i}$ are drawn.

We now provide the basic conditions needed to handle the sieve component of
the estimator. In accordance with the definition of a profiled likelihood,
all gradients with respect to $\theta $ below (denoted by $\nabla $)
incorporate the effect of simultaneous changes in the nuisance parameters
through the function $\omega \left( \theta \right) $ or $\hat{\omega}\left(
\theta \right) $. This approach provides a simple way to formally abstract
away the nuisance parameters from the expansion relevant to the asymptotics
of $\hat{\theta}$. Let $\mathcal{X},\mathcal{Y},\mathcal{Z}$,$\mathcal{W}$
denote the support of $X,Y,Z,W$, respectively\footnote{%
The assumption of rectangular support of $\left( X,Y,Z,W\right) $ is made
purely for notational convenience and can be trivially relaxed.}, while $%
\Theta $ is the parameter space for $\theta $. Let $\mathcal{F}$ denote a
neighborhood of $f_{0}$ (where the sup-norm is used for $f$). With these
definitions in mind, we can now state our key assumptions.

\begin{assumption}[Consistency]
\label{asscons}(i) $\mathcal{L}\left( \theta ,f_{0}\right) $ is uniquely
maximized at $\theta =\theta _{0}$ for $\theta _{0}$ in the interior of $%
\Theta $ with $\Theta $ compact, (ii) $\sup_{\theta \in \Theta }\linebreak[3]%
\sup_{f\in \mathcal{F}}\left\vert \mathcal{\hat{L}}\left( \theta ,f\right) -%
\mathcal{L}\left( \theta ,f\right) \right\vert \overset{p}{\longrightarrow }%
0 $, and (iii) $\mathcal{L}\left( \theta ,f\right) $ is continuous in $f$ at 
$f_{0}$ uniformly for $\theta \in \Theta $.
\end{assumption}

\begin{assumption}[Limiting distribution]
\label{asssieve} (i) $\sup_{\theta \in \Theta ,f\in \mathcal{F}}\left\vert
\nabla \nabla ^{\prime }\mathcal{\hat{L}}\left( \theta ,f\right) -\nabla
\nabla ^{\prime }\mathcal{L}\left( \theta ,f\right) \right\vert \overset{p}{%
\longrightarrow }0$ (ii) $H=\nabla \nabla ^{\prime }\mathcal{L}\left( \theta
_{0},f_{0}\right) $ is invertible, (iii) $\nabla \nabla ^{\prime }\mathcal{L}%
\left( \theta ,f\right) $ is continuous in $f$ at $f_{0}$ uniformly for $%
\theta \in \Theta $ and (iv) $\nabla \nabla ^{\prime }\mathcal{L}\left(
\theta ,f_{0}\right) $ is continuous in $\theta $ at $\theta _{0}$.
\end{assumption}

%(i) $\sup_{\theta \in \Theta }\left\vert \nabla \nabla^{\prime }\mathcal{\hat{L}}\left( \theta ,f_{0}\right) -\nabla \nabla^{\prime }\mathcal{L}\left( \theta ,f_{0}\right) \right\vert \overset{p}{%\longrightarrow }0$
We deliberately phrase Assumptions \ref{asscons} and \ref{asssieve} in a
high-level form because they arise in the asymptotic analysis of a
conventional sieve MLE estimator and a number of different possible
sufficient conditions are already available in the literature (e.g., %
\citet{hu2008instrumental}, \citet{chen:hb}).
For similar reasons, we remain
agnostic about the underlying location sampling process. As long as a given
interval of values for $\Delta s$ is repeatedly sampled as sample size grows
- allowing density estimation - our analysis is compatible with various
schemes, possibly featuring both infill or increasing domain asymptotics.
Assumption \ref{asscons}(i)
merely restates the conclusion of our earlier identification argument.
Assumptions \ref{asscons}(ii) and \ref{asssieve}(i) only require uniform
consistency and thus follow from uniform laws of large numbers for spatial
data (see \citet{jenish2009central}, who establish laws of large numbers
under mixing and moment conditions and turn them into uniform laws of large
numbers by adding stochastic equicontinuity and dominance). These conditions
are slightly strengthened here (relative to a standard sieve MLE) to account
for an estimated $f_{0}$. Assumptions \ref{asscons}(iii), \ref{asssieve}(ii)
and (iii) do not involve random quantities, hence the spatial nature of the
data is of no consequence. Assumption \ref{asscons}(iii) and \ref{asssieve}%
(iii) ensures that estimation of $f_{0}$ will not degrade the estimator's
properties.

We now use a more primitive formulation for the assumptions that are
specific to our estimator, for instance, those related to the fact that the
distribution of $Z$ is estimated and that the $Z_{i}$ are simulated draws.

\begin{assumption}[Support]
\label{asssupp}(i) $f_{Y,X|W}\left( y,x|w;\kappa \right) \geq \varepsilon >0$
$\forall x,y,w\in \mathcal{X}\times \mathcal{Y}\times \mathcal{W}$ and for $%
\kappa $ in a neighborhood $\mathcal{K}$ of $\kappa _{0}$ and (ii) $\mathcal{%
X},\mathcal{Y},\mathcal{Z},\mathcal{W}$ are compact.
\end{assumption}

To ensure that an estimator of the conditional density $f_{Z|Y,X,W}\left(
z|y,x,w;\kappa \right) $ of interest is well-behaved, it is natural to
require a nonvanishing conditioning density $f_{Y,X|W}\left( y,x|w;\kappa
\right) $. (The conditioning on $w$ is kept because we model the covariate $w
$ parametrically.) This type of \textquotedblleft nonvanishing
denominator\textquotedblright\ assumption is commonly made in the analysis
of semiparametric estimators with estimated densities, but could be relaxed
through tail trimming arguments. We do not consider this extension here,
because the intricate details needed would obscure the main ideas.

We also need to specify the nature of the spatial dependence of the
variables.

\begin{assumption}[Spatial dependence]
\label{assmix}The process $\left( X\left( s\right) ,Y\left( s\right)
,W\left( s\right) \right) $, indexed by $s\in \mathcal{S}$, is stationary
and strongly mixing.
\end{assumption}

This can be relaxed to the weaker mixing conditions given in %
\citet{carbon1997kernel}, but this extension is not spelled out here for
conciseness. Next, the use of kernel estimation is associated with some
familiar assumptions regarding the kernel and the smoothness of the
densities involved.

\begin{assumption}[Kernel]
\label{asskern2}The kernel $K_{yxz}$ is of dimension $d=2d_{x}+d_{y}+d_{s}$,
of order\footnote{%
See Definition \ref{defkorder} in Appendix \ref{appproof} for a formal statement.} $r>d$, satisfies a uniform Lipschitz condition and is
bounded. The bandwidth is selected such that $h=Cn^{-\varepsilon -1/\left(
2r\right) }$ for some $C,\varepsilon >0$. The kernel $K_{yx}$ and associated
bandwidth satisfy a similar assumption with $d=d_{x}+d_{y}+d_{s}$.
\end{assumption}

For simplicity the bandwidth is taken to decay at the same rate along all dimensions, while different prefactors along these dimensions can simply be incorporated in the definition of the kernel function itself.

\begin{assumption}[Density smoothness]
\label{assdensmo}The density $f_{\tilde{Y},\tilde{X},\tilde{Z}}\left( \tilde{y},\tilde{x},\tilde{z};\Delta s\right) $ is $r$ times
uniformly continuously differentiable in all its arguments.
\end{assumption}

Given the index structure, these standard kernel assumptions are augmented
by constraints on the link functions that ensure that the properties of the
estimated density of $\tilde{y},\tilde{x},\tilde{z}$ carry over to that of $y,x,z$.

\begin{assumption}[Link function]
\label{asslink}(i) The link function $G\left( (\tilde{y},\tilde{x},\tilde{z}),w,\kappa
\right) $ is one-to-one in its first argument with
uniformly nonsingular Jacobian. (ii) The inverse of $G\left( (\tilde{y},\tilde{x},\tilde{z}),w,\kappa \right) $ with respect to the first argument is uniformly
continuously jointly differentiable twice in both $\kappa $ and $(\tilde{y},\tilde{x},\tilde{z})$.
\end{assumption}

As the estimator is semiparametric in nature, its asymptotics will depend on
various score functions which need to be sufficiently smooth to
asymptotically eliminate any bias. The following assumption could be phrased
slightly more primitively by explicitly expanding the gradient $\nabla $
with respect to $\theta $, but at the expense of distracting notational
complications.

\begin{assumption}[Score smoothness]
  \label{asssemipar}The expected score
\begin{equation*}
  E\left[ \nabla \ln L\left( G\left( (\tilde{y},\tilde{x},\tilde{z}),W,\kappa _{0}\right) ,W;\theta _{0},\omega \left( \theta _{0}\right) \right) \right]
\end{equation*}
is $r$ times uniformly continuously differentiable in $\tilde{x},\tilde{y}$ (and $\tilde{z}$).
\end{assumption}

The following conditions are needed to account for the simulated nature of $ Z_{i}$ and are simple to verify by inspection.

\begin{assumption}[Generated $Z_{i}$]
\label{asszhat}(i) $\nabla \ln L\left( y,x,z,w;\theta ,\omega \left( \theta
\right) \right) $ is bounded and Lipschitz in $z$ and (ii) $f_{Z|YXW}\left(
z|y,x,w\right) $ is bounded and bounded away from zero over its support.
\end{assumption}

Finally, the preliminary estimator $\hat{\kappa}$ needs to satisfy some
simple condition associated with root $n$ consistency and asymptotic
normality.

\begin{assumption}[Estimator $\hat{\protect\kappa}$]
\label{assnuinf}The estimator $\hat{\kappa}$ admits the asymptotically
linear representation $n^{1/2}\left( \hat{\kappa}-\kappa _{0}\right)
=n^{-1/2}\sum_{i=1}^{n}\psi _{\kappa }\left( Y_{i},X_{i},Z_{i},W_{i}\right)
=O_{p}\left( 1\right) $\ for some known influence function $\psi _{\kappa
}\left( y,x,z,w\right) $.
\end{assumption}

This condition is easy to satisfy, as the estimator $\hat{\kappa}$ would
typically consists of regressions (e.g. of $Y$ on $W$, of $X$ on $W$ and of $%
Z$ on $W$) or, more generally, could be a GMM estimator.

We are now ready to state our main asymptotic result (established in Appendix \ref{appproof}):

\begin{theorem}[Asymptotically linear representation]
\label{thasslin}Under Assumptions \ref{assindex}, \ref{asscons}, \ref%
{asssieve}, \ref{asssupp}, \ref{assmix}, \ref{asskern2}, \ref{assdensmo}, %
\ref{asslink}, \ref{asssemipar}, \ref{asszhat}, \ref{assnuinf} , 
\begin{eqnarray*}
n^{1/2}\left( \hat{\theta}-\theta _{0}\right)  &=&n^{-1/2}\sum_{i=1}^{n}\psi
_{\text{MLE}}\left( Y_{i},X_{i},Z_{i},W_{i}\right)+n^{-1/2}\sum_{i=1}^{n}\psi _{\text{kernel}}\left(
Y_{i},X_{i},Z_{i},W_{i}\right)  \\
&&+n^{-1/2}\sum_{i=1}^{n}\psi _{\text{cov}}\left(
Y_{i},X_{i},Z_{i},W_{i}\right) +o_{p}\left( 1\right) 
\end{eqnarray*}%
where 
\begin{equation*}
\psi _{\text{MLE}}\left( y,x,z,w\right) =-H^{-1}\nabla \ln L\left(
y,x,z,w;\theta _{0},\omega \left( \theta _{0}\right) \right) 
\end{equation*}%
is the usual influence function of a standard sieve semiparametric MLE with
observed $Y_{i},X_{i},Z_{i},W_{i}$, while 
the correction term due to constructing the measurement $Z_{i}$ is
\begin{eqnarray*}
\psi _{\text{kernel}}\left( y,x,z,w\right)  &=&H^{-1} \left( E\left[ \nabla \ln \tilde{L}_{1}\left( G^{-1}\left( (Y,X,Z), W \right)  ;\theta _{0}\right) \right] \right. 
-\nabla \ln \tilde{L}_{1}\left( G^{-1}\left( (Y,X,Z), W \right)  ;\theta _{0}\right)  \\
&&+ E\left[ \nabla \ln \tilde{L}_{2}\left( G^{-1}\left( (Y,X,Z), W \right)  ;\theta _{0}\right) \right] 
\left. -\nabla \ln \tilde{L}_{2}\left( G^{-1}\left( (Y,X,Z), W \right)  ;\theta _{0}\right) \right) 
\end{eqnarray*}%
where
\begin{eqnarray*}
\nabla \ln \tilde{L}_{1}\left( G^{-1}\left( (Y,X,Z), W \right)  ;\theta _{0}\right) &=& E\left[ \nabla \ln L\left( G\left( (\tilde{y},\tilde{x},\tilde{z}%
),W\right) ,W;\theta _{0},\omega \left( \theta _{0}\right) \right) \right]\\
\nabla \ln \tilde{L}_{2}\left( G^{-1}\left( (Y,X,Z), W \right)  ;\theta _{0}\right) &=& E\left[ \nabla L_{2}\left( G_{yx}\left( (\tilde{y},\tilde{x},\tilde{z}%
),W\right) ,W,\theta _{0}\right) \right]
\end{eqnarray*}
in which $G_{yx}\left( (\tilde{y},\tilde{x},\tilde{z}),w\right) $ denotes the $y
$ and $x$ elements of the vector $G\left( (\tilde{y},\tilde{x},\tilde{z}%
),w\right) $ and
\begin{equation*}
\nabla L_{2}\left( y,x,w,\theta _{0}\right) =\int f_{Z|YXW}\left(
z|y,x,w\right) \nabla \ln L\left( y,x,z,w;\theta _{0},\omega \left( \theta
_{0}\right) \right) dz.
\end{equation*}%
Finally,
\begin{equation*}
\psi _{\text{cov}}\left( y,x,z,w\right) =-H^{-1}E\left[ \nabla _{\kappa
}^{\prime }\ln f\left( Z|Y,X,W;\kappa _{0}\right) \nabla \ln L\left(
Y,X,Z,W;\theta _{0},\omega \left( \theta _{0}\right) \right) \right] \psi
_{\kappa }\left( y,x,z,w\right) .
\end{equation*}%
is the correction term for the estimation of the nuisance parameter $\kappa $%
, in which $\psi _{\kappa }\left( y,x,z,w\right) $ denotes the influence
function of the estimator $\hat{\kappa}$ from Assumption \ref{assnuinf}.
\end{theorem}

The conclusion of Theorem \ref{thasslin} is stated in a way such that any
central limit theorems for sample averages involving spatial data (see,
e.g., 
\citet{bolthausen1982central, lahiri2003central, jenish2009central,
jenish2012spatial} for CLT under various types of mixing and moment
conditions) can be freely used to obtain the limiting distribution. If a
resampling approach is preferred, a block bootstrap 
\citep{hall1995blocking,nordman2007optimal,carlstein:blockboot}
approach should be used to account
for the possible spatial dependence. While the existence of an
asymptotically linear representation is the key result that enables
bootstrap validity, formally establishing this requires additional technical
conditions. The simplest of such conditions would be that our Assumptions %
\ref{asscons}-\ref{assnuinf} hold with sample quantities replaced by
bootstrap versions and the population quantities replaced by a sequence of
sample quantities (e.g., paralleling Theorem 23.5 in \citet{vaart:asymp}).
More primitive conditions that do not involve bootstrapped quantities\ can
also be stated in terms of standard results (e.g. \citet{bickel:boot}, as
summarized in Theorem \ref{thbickel} in Appendix \ref{appproof}).

\section{Simulations}

We conduct simulations to assess the performance of our measurement error robust estimator. We generate a correlated Gaussian random field $X^{*}(S)$ on a rectangle subset (130$\times$65) of $\mathbb{R}^2$, on which we then construct $Y(S) = g(X^{*}(S)) + U(S)$ and $X(S) = X^{*}(S) + V(S)$ which are observed at random locations $S = S_{i}, \ i=1,...,n$.

Specifically, $X^*$ has a normal distribution with mean $3.5$, variance $1$, and correlation to its first-order neighbors of 0.6. The correlation is roughly divided by $3$ with each increment in distance. We specify $g(x^*) = \theta_1 + \theta_2 x^*$, and $(\theta_1, \theta_2) = (-3.5, 2)$. The error terms, $U$ and $V$, are normally distributed independently of $X^{*}$ with standard deviations of $1.3$ and $0.8$, respectively. 

A realization of the underlying random fields is depicted in Figure 1. 

\begin{figure}[!ht]
\begin{subfigure}{.35\textwidth}
  \centering
  \includegraphics[width=.9\linewidth]{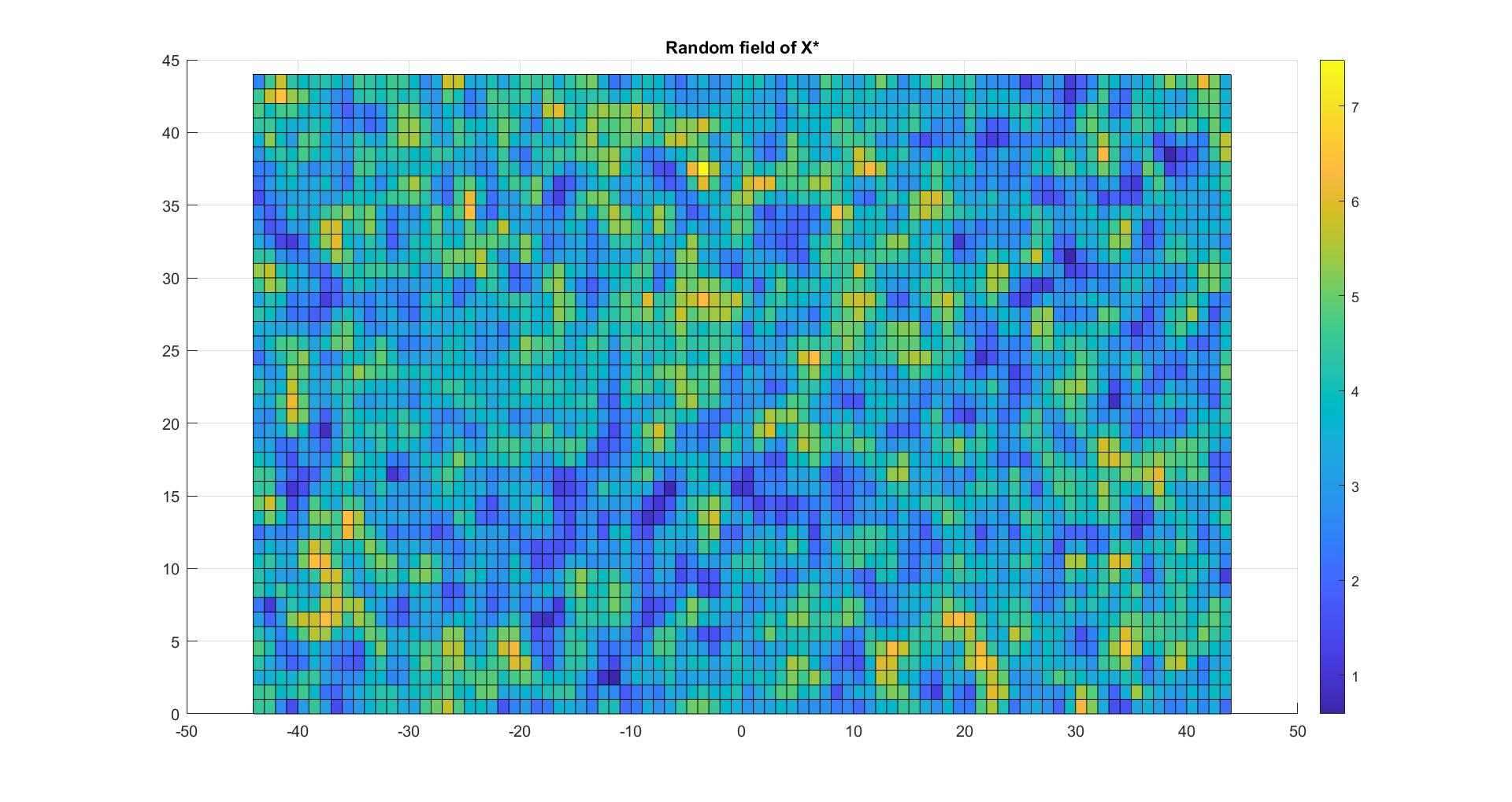}
\end{subfigure}%
\begin{subfigure}{.35\textwidth}
  \centering
  \includegraphics[width=.9\linewidth]{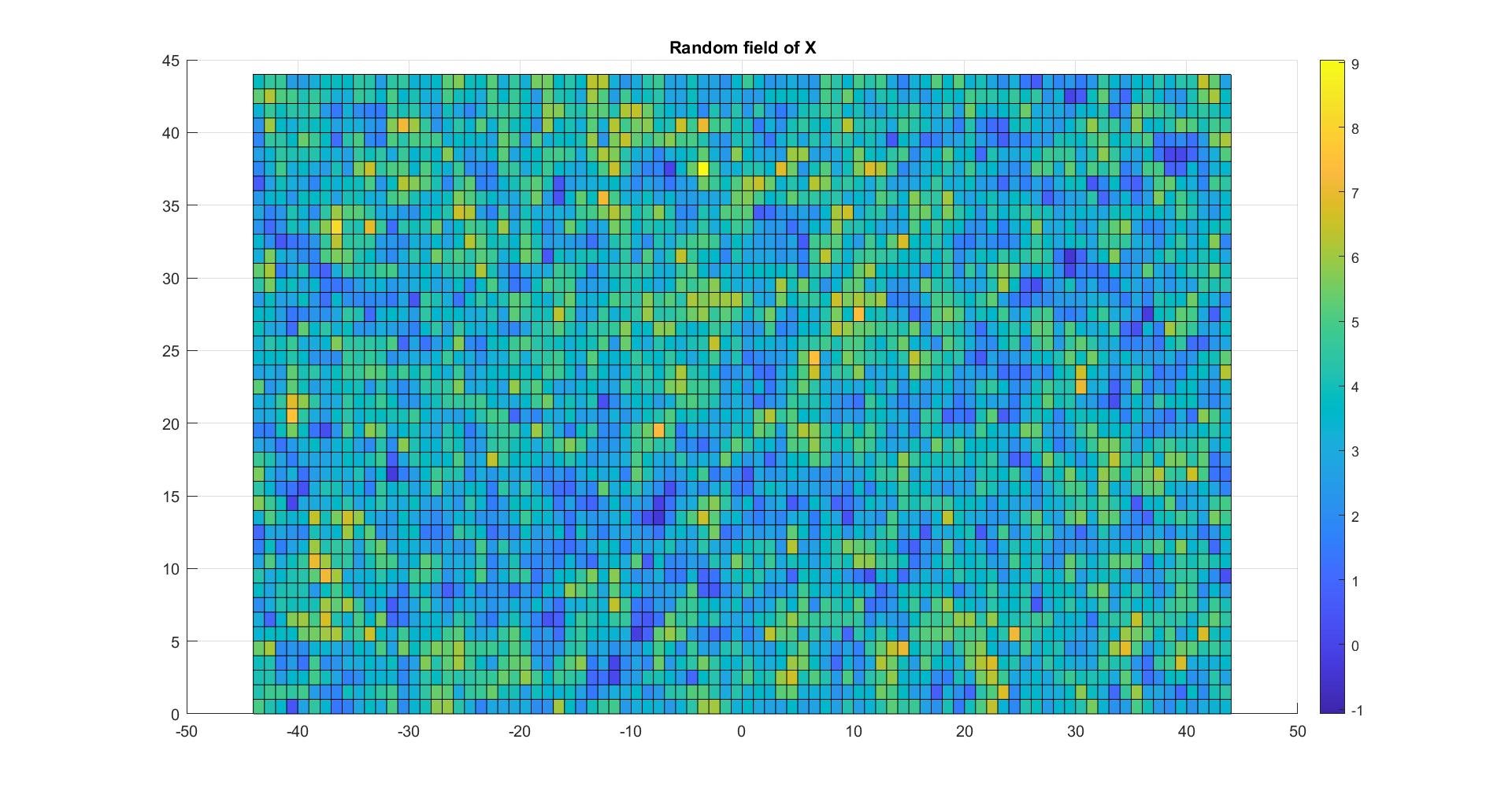}
\end{subfigure}%
\begin{subfigure}{.35\textwidth}
  \centering
  \includegraphics[width=.9\linewidth]{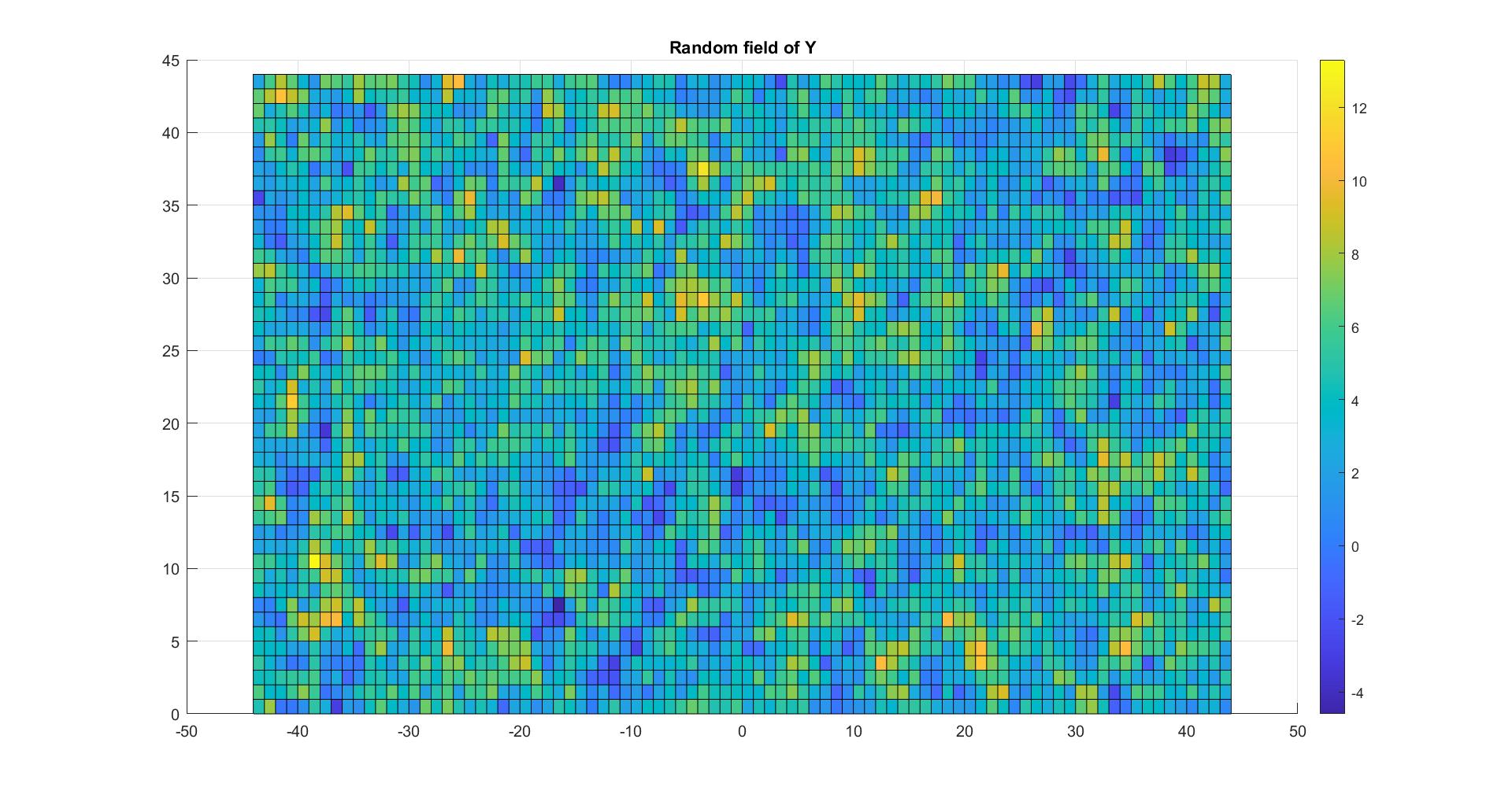}
  \label{fig:sfig3}
\end{subfigure} \\
\caption{Spatial plot of a realization of the underlying random fields. \\ Leftmost: $X^*$; middle: $X$; rightmost: $Y$.}
\end{figure}

We parametrically specify $f(y_{i}|x^{*})$ in the optimization problem and analyze results for $(\theta_1, \theta_2, \sigma_u)$. Two forms of the estimators are tested: a simple, unweighted average over all distances, and the optimal inverse-variance-weighted average. In both cases, the density $f_{Y(s), X(s), X(s+\Delta S)}(y,x,z;\Delta s)$ is estimated by adaptive kernel density estimation based on diffusion processes \citep{botev2010kernel}, which provides a bandwidth selection rule and allows us to construct $f_{X(s+\Delta S)|Y(s), X(s)}(z | y,x;\Delta s)$ and sample pseudo-instruments. 

The number of sieve terms has been chosen by examining the resulting densities and ensuring small variations in the number of sieves do not cause the resulting estimator to vary much. This is in line with the suggestion in \citet{hu2008instrumental} "that a valid smoothing parameter can be obtained by scanning a range of values in search of a region where the estimates are not very sensitive to small variations in the smoothing parameter".
With Section 3's notations, this leads to $i_n = 6$, $j_n=4$ for $f(x|x^*)$, $i_n = 4$, $j_n=4$ for $f(z|x^*)$, and $i_n = 4$ for $f(x^*)$. In appendix \ref{appaddsim}, we report additional simulations for a range of sieves truncation choices -- $i_n=j_n=2k, k=1, 2, 3$ for all densities -- which suggest performance does not depend strongly on the number of sieves terms within the range $4-6$. Lower values appear too rough and reduce the performance of the estimator, while higher values add too much variance and let the number of parameters explode, which also increases computational burden.
These two versions of the estimators are compared to the infeasible OLS that uses the unobserved regressor and to the biased OLS estimator that regresses on the mismeasured regressor. Results are displayed in Tables \ref{tabres1} through \ref{coverage}. 

\begin{table}[!ht]\caption{Simulation Results}
\[\begin{array}{ | l | l | l | l |}
\hline
\text{Parameter }\theta_1\text{, true value:}-3.5 & \mbox{Mean} & \mbox{Standard deviation} & \mbox{RMSE} \\ \hline
	\mbox{Infeasible OLS} & -3.50 & 0.12 & 0.12 \\ \hline
	\mbox{OLS} & -0.77 & 0.15 & 2.74 \\ \hline
        \mbox{IV Nearest-Neighbor} & -3.56 & 0.34 & 0.35 \\ \hline
	\mbox{Unweighted Spatial} &-3.54 & 0.19 & 0.19 \\ \hline
	\mbox{Weighted Spatial}  &-3.53 & 0.19 & 0.19 \\ \hline
\end{array}\]
%\end{table}\begin{table}[!ht]
\[\begin{array}{ | l | l | l | l |}
\hline
\text{Parameter }\theta_2\text{, true value: }\phantom{-}2.0	& \mbox{Mean} & \mbox{Standard deviation} & \mbox{RMSE} \\ \hline
        \mbox{Infeasible OLS} & 2.00 & 0.03 & 0.03 \\ \hline
	\mbox{OLS} & 1.22 & 0.04 & 0.78 \\ \hline
        \mbox{IV Nearest-Neighbor} & 2.03 & 0.10 & 0.10 \\ \hline
        \mbox{Unweighted Spatial} & 2.03 & 0.05 & 0.06 \\ \hline
	\mbox{Weighted Spatial} & 2.03 & 0.05 & 0.06 \\ \hline
\end{array}\]
%\end{table}\begin{table}[!ht]
\[\begin{array}{ | l | l | l | l |}
\hline
\text{Parameter }\sigma_u\text{, true value: }\phantom{-}1.3	& \mbox{Mean} & \mbox{Standard deviation} & \mbox{RMSE} \\ \hline
	\mbox{Infeasible OLS} & 1.30 & 0.02 & 0.02 \\ \hline
	\mbox{OLS} & 1.80 & 0.03 & 0.50 \\ \hline
        \mbox{IV Nearest-Neighbor} & 1.18 & 0.11 & 0.17 \\ \hline
	\mbox{Unweighted Spatial} & 1.19 & 0.07 & 0.13 \\ \hline
	\mbox{Weighted Spatial} & 1.20 & 0.06 & 0.12 \\ \hline
\end{array}\]
\caption*{Simulations with 1500 observations. Infeasible OLS refers to the infeasible OLS estimator using the unobserved true regressor; OLS is the biased OLS estimator using mismeasured covariate; Nearest-Neighbor is the standard IV estimator using the nearest neighbor covariate value as instrument; Unweighted Spatial is our unweighted average spatial estimator; Weighted Spatial is the optimally weighted average spatial estimator.}\label{tabres1}
\end{table}

It is seen --- as expected in presence of substantial measurement error --- that the biased OLS regression using the mismeasured regressor performs poorly, displaying strong attenuation bias.  

Our estimator exhibits significant improvement over OLS in terms of bias. For all parameters, it also attains a Root Mean Square Error (RMSE) that is not much larger than that of the efficient, infeasible, OLS estimator that uses the actual covariate. Of course, a finite-sample bias is expected, especially given the slight misspecification induced by the truncation of the sieve expansion, but its magnitude remains reasonable and far lower than OLS's. Furthermore, the estimator performs similarly to OLS in terms of variance.

Another informative benchmark is the use of the nearest neighbor as an instrument. Because of the linear specification and classical error, this is an alternative to our method in this specific case.
While the instrumental variable approach provides a substantial improvement over inconsistent OLS, our estimator improves accuracy even further. The randomness in nearest neighbor distances is likely the source of the additional variability in the IV estimator that is not present in our approach. Hence, there is little or no cost to using our method over a nearest-neighbor instrument in the linear case.

Note, however, that our method's key advantage is its broader applicability. Indeed, the nearest neighbor IV estimator will generally be inconsistent in nonlinear models since traditional instruments cannot serve as repeated measurements in nonlinear models \citep{Amem:IV}. 
Furthermore, nearest neighboring observations of $X$ cannot play the role of a the variable ``$Z$'' in the \citet{hu2008instrumental} framework, due to the randomly-chosen, unevenly-spaced locations $S$. Our approach avoids these problems.

To illustrate the inconsistency problem associated with using nearest neighbor, we now consider nonlinear models. Letting $g$ be a third-order polynomial with coefficients $\theta = (-3.5, 0.2, 0.2, -0.05)$, we compare the results from applying our method to the nearest-neighbor IV approach. For ease of interpretation, we report the value of the function $g$ at quartiles of the distribution of $x^*$, rather then the coefficients.

The results are displayed in Table \ref{polynomial}. Using the nearest-neighbor as an instrument clearly is an unreliable strategy: the curve is very noisily estimated and systematically biased, especially for values of $x^*$ away from its median. By contrast, our method is able to recover the true curve with minimal bias and good accuracy. 
\begin{table}[!ht]\caption{Polynomial}\label{polynomial}
\[\begin{array}{ | l | r | r | r | }
\hline
\mbox{Quartile} & 1^\text{st} & 2^\text{nd} & 3^\text{rd} \\
\hline
x^* & 2.82 & 3.51 & 4.17 \\
g(x^*) & -2.47 & -2.50 & -2.81 \\
\hline
\mbox{Mean weighted-spatial} & -2.44 & -2.55 & -3.00 \\
\mbox{Mean IV-NN} & -0.73 & -2.79 & -4.80 \\
\hline
\mbox{Std. weighted-spatial} & 0.14 & 0.16 & 0.15 \\
\mbox{Std. IV-NN} & 17.11 & 5.32 & 22.46 \\
\hline
\mbox{RMSE weighted-spatial} & 0.14 & 0.16 & 0.24 \\
\mbox{RMSE IV-NN} & 17.16 & 5.32 & 22.49 \\
\hline
\end{array}\]
%\[\begin{array}{ | l | l | l | l | l | l | }
%\hline
%x^* & 0.12 & 2.82 & 3.51 & 4.17 & 7.65 \\
%g(x^*) & -3.47 & -2.47 & -2.50 & -2.81 & -12.65 \\
%\hline
%\mbox{Mean weighted-spatial} & -3.48 & -2.44 & -2.55 & -3.00 & -14.70 \\
%\mbox{Mean IV-NN} & -15.29 & -0.73 & -2.79 & -4.80 & 26.21 \\
%\hline
%\mbox{Std. weighted-spatial} & 0.30 & 0.14 & 0.16 & 0.15 & 1.65 \\
%\mbox{Std. IV-NN} & 120.07 & 17.11 & 5.32 & 22.46 & 367.79 \\
%\hline
%\mbox{RMSE weighted-spatial} & 0.30 & 0.14 & 0.16 & 0.24 & 2.63 \\
%\mbox{RMSE IV-NN} & 120.35 & 17.16 & 5.32 & 22.49 & 368.92 \\
%\hline
%\end{array}\]
\caption*{Performance of our method (weighted-spatial) vs. using the nearest-neighbor as an instrument (IV-NN) in recovering $g(x^*)$ at quartiles of the distribution of $x^*$.}
\end{table}
% Nearest-neighbor as an instrument: bias/std/rmse
%theta_1
%-3.56	0.34	0.35
%theta_2
%2.03	0.10	0.10
%sigma
%1.18	0.11	0.17

%The optimally-weighted average performs slightly better than the unweighted average in terms of RMSE, with the improvement originating from a bias reduction. While modest in the above simulations, the gains are more apparent if weaker instruments are included by considering less relevant values of $\Delta s$ - weak instruments - since an unweighted average does not downplay their impact on the resulting estimator. Gains can also depend on the nature of the identification assumption; other simulations suggest a greater RMSE discrepancy between weighted and unweighted estimators when identification is based on the mode of the conditional distribution rather than the conditional expectation.   

Perhaps more importantly, the method's low bias enables it to also deliver reliable confidence intervals with bootstrapped standard errors. The (spatial block-) bootstrap procedure is described in Appendix C, with $22 \times 15$ blocks. To speed up calculations, any sample average is pre-computed on each block once, and then any block bootstrap sample averages can be quickly computed from averages over the pre-computed block averages.
%Because of computational requirements and the underlying density estimation, we use the jacknife trick to cut computation and err on the conservative side of slightly large $22 \times 15$ blocks. 
% this sounded too mysterious

Table \ref{coverage} displays the coverage performance of 95\% confidence intervals for all estimators. It appears that our proposed estimator achieves coverage close to 95\% and does not fall far from the infeasible OLS estimators that makes use of the perfectly-measured regressor.

\begin{table}[!ht]\caption{Coverage}\label{coverage}
\[\begin{array}{ | l | l | l |}
\hline
	& \theta_1 & \theta_2 \\ \hline
	\mbox{Infeasible OLS} & 0.95 & 0.94 \\ \hline
    \mbox{OLS} & 0 & 0 \\ \hline
    \mbox{Unweighted spatial} & 0.96 & 0.92 \\ \hline
    \mbox{Weighted spatial} & 0.96 & 0.96 \\ \hline
\end{array}\]
\caption*{Coverage performance of 95\% confidence intervals.}
\end{table}

Further simulations exploring the link between the distance $\Vert \Delta s\Vert$ that determines the choice of instrument and estimation accuracy reveal a non-trivial relationship. Using the RMSE of the estimated $\theta_2$ as a figure of merit, we find that using $\Vert \Delta s\Vert = 1.5$ provides the best result with a RMSE of 0.04, beating both the closer distance of 0.75 (RMSE 0.07) and the larger distance of 1.25 (RMSE 0.09). The reason for this non-monotone behavior is likely that closer observations improve the instrument's strength, while larger distances induce a higher count of observations, which allows a more precise estimate of the conditional density. The analysis of the $\Delta s$-dependence is done here for illustration purposes --- when using a weighted average over a range of $\Delta s$ (as we shall do in our application), there is no need to select a specific $\Delta s$.

While estimators from individual distances can exhibit heavy tails, the presence of outlier estimates is alleviated thanks to averaging over different estimates. For instance in estimating $\theta_2=2$, the first percentile is 1.98 and the $99^{th}$ reaches 2.08. In this example, the intercept is somewhat more prone to outliers; in estimating $\theta_1 = -3.5$, the corresponding percentile figures are -4.15 and -3.51. 

In Appendix \ref{appaddsim}, we report results from additional simulations. These are designed to further assess the robustness of the estimator to variations in our baseline setup, in particular the presence of non-classical measurement error or non-linearities in the outcome equation.
%First, we generate the non-classical error-contaminated observed regressor as $X=(X^*+2/3)B-1/3$ where $B$ (conditional on $X^*$) is a $\mbox{Beta}(X^* + 4/3$, 4/3)\footnote{In the very rare event that the argument of the beta random variable falls below $0$, we would truncate the variable to $0$.} and we implement our estimator with mode-centering.
First, we generate the non-classical error-contaminated observed regressor $X$ as a log-normal($\mu$, $\sigma^2$) where the log-normal (conditional on $X^*$) is has mean $\ln(X^*)$\footnote{In the extremely rare event that the argument of the logarithm falls below $0$, it would be truncated to $\ln(0.001)$.} and variance $1/25$ and we implement our estimator with median-centering.
We also consider nonlinear conditional expectations with a standard probit model. Table \ref{mode_nonlin_results} (bias, standard deviation, RMSE) and Table \ref{coverage_results} (coverage) in the Appendix \ref{appaddsim} display the results for these modified simulation designs. They suggest that our estimator performs adequately in such situations as well.

Finally, we investigate the impact of incorporating covariates and resorting to semiparametric estimation. A covariate is generated as $W=(X^* + \mathcal{N}(0, 1))^2/20$ -- and thus is correlated with the unobserved regressor -- and enters the outcome equation additively ($Y=W+g(X^*)+U$). We adopt a semiparametric procedure in which the covariate is allowed to shift location of densities. 
Table \ref{covariates_results} (bias, standard deviation, RMSE) and Table \ref{coverage_results} (coverage) report the results. They are in line with our baseline results, albeit with a general increase in RMSE.

\section{Application}
\label{secapp}

We revisit the influential study of \citet{michalopoulos2013pre} to demonstrate how our approach can effectively deliver measurement-error robust estimation and inference in the context of spatial data, without necessitating additional auxiliary variables, such as instruments or validation data. In this application, the possibility of significant measurement error in a key regressor is an important concern that existing methods have been unable to fully address.

This study investigates the relationship between pre-colonial ethnic political centralization and contemporary development. The underlying motivation is to confirm anecdotal observations that the pre-existence of a complex large-scale political structure within ethnic groups appears to strongly impact economic development, independently of political structures put in place during colonization. The pre-colonial political structure is captured by measures of the extent of 
jurisdictional hierarchy beyond the local level developed by \citet{murdock1967ethnographic}. Obtaining such measures is challenging, as it involves subjective assessments, and is thus prone to misclassification errors, as discussed by \citet{michalopoulos2013pre}. Since this quantity appears as a regressor in the analysis, the possibility of measurement error induced bias must be considered and we consider the observed regressor, $x_i$, to be a mismeasured version of the true regressor, $x_i^*$.

The dependent variable, $y_i$, in this study is economic activity. Given unavailability of comparable economic indicators across African ethnic homelands, the authors employ nighttime artificial light intensity as a proxy for economic activity, in the spirit of \citet{henderson2012measuring}, \citet{elvidge1997relation} and \citet{doll2006mapping}, among others.

Their main regression takes the following form: 
\begin{equation}
    y_i = \beta_0 + \beta_1 x_i^* + w_i' \beta_W + \epsilon_i
\end{equation}
where $y_i$ denotes light density at night, $x_i^*$ is the (correctly-measured) level jurisdictional hierarchy or ``complexity'', taking value in $\{0,1,2,3,4\}$, and $w_i$ is a vector of covariates including population density, location controls (distance to the capital city, distance to the border, and distance to the coast), geographic features (land suitability for agriculture, malaria stability index, land area under water, and petroleum and diamond dummies), and income per capita. Country fixed effects are also considered. 

Results from Table 2 and 3 in \citet{michalopoulos2013pre}, which are partially reproduced in Table \ref{tabmichal}, suggest that a one unit increase in the jurisdictional hierarchy index --- roughly corresponding to a one standard deviation increase --- leads to an increase in light luminosity of 20 \% (with all controls and country fixed effects) to 40\% (without controls) --- corresponding to a 0.1 to 0.2 standard deviation increase. See \citet{michalopoulos2013pre} for details. 

\begin{table}[!ht] \caption{Replicated results}
\label{tabmichal}
\[\begin{array}{ | l | l | l | l | l |}
\hline
	 & \mbox{Coefficient} & \mbox{se} & \mbox{95\% CI lb} & \mbox{95\% CI ub} \\ \hline
	\mbox{No controls} & 0.41 & 0.12 & 0.17 & 0.66  \\ \hline
	\mbox{Controls} & 0.2 & 0.05 & 0.1 & 0.29 \\ \hline
	\mbox{Controls and FE} & 0.18 & 0.05 & 0.08 & 0.27 \\ \hline
\end{array}\]
\caption*{OLS estimate for hierarchy index coefficient on (log) light luminosity; standard errors (se); lower bound (lb) and upper bound (ub) of 95\% confidence interval (CI). FE refers to country fixed effects.}
\end{table}

These results suggest a strong relationship between pre-colonial political complexity and current economic development, and here we seek to ensure that this finding is robust to the presence of misclassification errors. It is also of independent interest to quantify how prevalent classification errors are in such frameworks. We illustrate below how our methodology can deliver on these issues. \medskip

The spatial region comprises geographic locations indexed as an element of $\mathbb{R}^2$, with kilometers as unit of measure. We estimate the spatial autocorrelation of the hierarchical complexity to vary from 0.35 to 0.25 for distances between 10 and 150 kilometers, the complexity process being viewed as isotropic. This supports our identification strategy and we construct $\Vert \Delta s \Vert= j \times 10$ km  for $j =1, ..., 15$ as repeated measurements, after estimating the density through kernel smoothing, as in the simulations. 

Although the correctly classified variable is unobserved, it can be argued that the misclassification is at least mode-preserving \citep{schennach:hb}, i.e. for any true underlying level of complexity, the correct level is more likely
to be reported than any other (incorrect) level. We view the errors as some subjective mis-judgements unlikely to be spatially correlated and invoke assumption 2.1 with $d=0$ (though a robustness check avoiding the smallest $\Vert \Delta s\Vert$ delivers similar results). Combined with repeated measurements provided by the spatial structure, this identifies the distribution $f(y_{i}|x_i^{*})$.

In the sample, the highest level of complexity ($x_i=4$) occurs less than 1\% of the time, making it difficult to estimate probabilities involving the associated event accurately. To alleviate the issue, we pool outcomes $X=3$ and $X=4$ together at the value 3.\footnote{Alternative strategies would be to use the weighted-average value (3.1) or to drop observations with a 4. These options do not materially change the results, as expected by the very low frequency of 4s.} 

While our analysis allows for covariates by considering conditional densities, this particular application poses additional complications due to the large number of controls (more than ten before the inclusion of fixed effects). To account for these numerous covariates $W$, we adapt our nonparametric estimation strategy in the crucial steps of a) the estimation of the distribution of $z$ conditional on the observations to generate pseudo-instruments, and b) the maximum likelihood estimation step. 

For step a), we adopt the link function-based strategy embodied in Assumption \ref{assindex}, with a link function constructed element-by-element as follows.
We first decompose the joint distribution as $f(z,x,y |w) = f(y|w,x,z) f(x,z|w)$.
%?? is this correct:
(i) We then obtain a kernel density estimator on the conditional density of $y$ for all $(x, z)$ under the assumption that dependence on the covariates $w$ takes the form of a location shift, i.e. $y=G_y(x,z,\tilde{y})+w'\kappa_y$, where the kernel approach allows the link function $G_y(x,z,\tilde{y})$ to be nonparametric, while $\kappa_y$ is an unknown parameter and $\tilde{y}$ is a noise term independent from $w$.
(ii) We specify $\begin{pmatrix}z \\ x \end{pmatrix} =t\left[G_{zx}\left(\begin{pmatrix}(w'\kappa_{z}+\tilde{z} \\ (w'\kappa_{x}+\tilde{x} \end{pmatrix}\right)\right]$, where $G_x$ and $G_z$ are known link functions and $t$ is a truncation function that maps to the discrete support of $x$ or $z$ while $\kappa_x,\kappa_y$ are parameters to be estimated by maximum likelihood.
The cutoffs defining the function $t$ are treated as unknown and optimized over, while $G$ is taken to be linear and the noises $\tilde{x},\tilde{z}$ are zero-mean Gaussians with unknown variance matrix; they are independent from $w$. 
(iii) Finally, we recover the conditional distribution of $z$ from the joint and sample pseudo-instruments.

Regarding step b), our log-likelihood, once conditioned on covariates, takes the form $f_{xyz\vert w} = \sum_{x^*=0}^3 f_{y|x^*, w} f_{x|x^*} f_{z|w} f_{x^*|w}$ under the measurement error assumption $f_{x|x^*,w}=f_{x|x^*}$.
Since $z$ and $x^*$ are discrete with few support points and covariates have high dimension, we adopt the link function modeling strategy again for both $f_{z|w}$ and $f_{x^*|w}$ while $f_{x|x^*}$ is left fully nonparametric -- $f_{x|x^*}$ is a 4x4 matrix to estimate with only restrictions that probabilities are nonnegative and sum up to 1, and the mode-centering restriction.

Finally, we obtain bootstrap standard errors as in the simulations, which provided adequate coverage. We choose blocks of size $1250\times 1250$, which comprises less than 10\% of the sample per block as in the simulations.

Applying our measurement-error robust (inverse variance-weighted) spatial estimator yields the results shown in Table \ref{tabmere}.

\begin{table} \caption{Measurement error robust estimates}
\[\begin{array}{ | l | l | l | l | l |}
\hline
	 & \mbox{Coefficient} & \mbox{se} & \mbox{95\% CI lb} & \mbox{95\% CI ub} \\ \hline
	\mbox{No control} & 1.82 & 0.19 & 1.44 & 2.20 \\ \hline
	\mbox{Controls} & 1.31 & 0.29 & 0.74 & 1.88 \\ \hline
	\mbox{Controls and FE} & 0.91 & 0.37 & 0.18 & 1.64 \\ \hline
\end{array}\]
\caption*{Measurement-error corrected estimate for hierarchy index coefficient on (log) light luminosity; standard errors (se) are estimated using a block bootstrap; lower bound (lb) and upper bound (ub) of 95\% confidence interval (CI). FE refers to country fixed effects.}
\label{tabmere}
\end{table}

A regression without additional controls yields a statistically significant estimate of 1.80, a much larger finding than that of the OLS estimator. The coefficient decreases as controls are added, though measurement error robust estimates still point to a stronger influence of political complexity on development than the biased OLS coefficients do.
The use of our measurement error robust estimator also does not come at the cost of statistical significance, as all coefficients remain statistically significant.

Our method also identifies the misclassification matrix, which is reported in tables \ref{tabmismat1}, \ref{tabmismat2}, and \ref{tabmismat3}.\footnote{We report the estimate with the instrument constructed at a distance of $10$km rather than the weighted averages to ensure probabilities sum up to 1. The matrix looks broadly similar for instruments coming from  different distances.} For reference, the estimates of $\mathbb{P}[X=i]$ are 0.26, 0.39, 0.24, and 0.11 for $i=0, 1, 2, 3$, respectively. There appears to be substantial misclassification, regardless of the specification. While extreme misclassifications are less frequent, subjective assessments can often deviate to nearby categories which is reflected in these estimates.

The large differences in the estimated coefficients
between Tables \ref{tabmichal} and \ref{tabmere} further suggests that measurement error could have important effects on the coefficients of interest in this application. This prompts us to formally test if the effect of measurement error is statistically significant. Consider an estimator that is robust to
measurement error ($\hat{\beta}$) and one that is not ($\tilde{\beta}$),
with corresponding influence functions $\hat{\psi}$ and
$\tilde{\psi}$ and corresponding standard errors $\hat{\sigma}$
\ and $\tilde{\sigma}$. A formal test that the presence of
measurement error affects the results can be based on the statistic:%
\[
\frac{\left\vert \hat{\beta}-\tilde{\beta}\right\vert }{\sqrt{\hat{E}\left[
\left( \hat{\psi}-\tilde{\psi}\right) ^{2}\right] }}
\]%
where we consider scalar $\beta $ for simplicity and where $%
\hat{E}\left[ \ldots \right] $ denotes sample averages. To avoid
computing the correlation between the influence functions, consider the
worst-case scenario where the two influence functions are perfectly
negatively correlated: $\tilde{\psi}=-c\hat{\psi}$ for some $c>0$.
The denominator can thus be bounded as:%
\[
\sqrt{E\left[ \left( \hat{\psi}-\tilde{\psi}\right) ^{2}\right] } \leq 
\sqrt{E\left[ \left( \hat{\psi}+c\hat{\psi}\right) ^{2}\right] }=\left(
1+c\right) \sqrt{E\left[ \hat{\psi}^{2}\right] }=\left( 1+c\right) \hat{%
\sigma}=\hat{\sigma}+\tilde{\sigma}.
\]
A valid (but conservative) test can thus be based on comparing the
ratio $\left\vert \hat{\beta}-\tilde{\beta}\right\vert /\left( \hat{\sigma}+%
\tilde{\sigma}\right) $ to standard normal critical values. For the 3
estimates reported in Tables \ref{tabmichal} and \ref{tabmere}, this statistic takes the values $4.5$
, $2.7$ and $1.1$, thus indicating a statistically
significant effect (at 95\% level) of measurement error in most
cases. Overall, the results support the view that measurement error is a major concern in such applications and that our method offers a viable avenue to address this issue. 

\begin{table}[!ht] \caption{$\mathbb{P}[X=i | X^*=j]$ (no control)}
\[\begin{array}{ | l | l | l | l | l | l |}
\hline
     & j=0 & j=1 & j=2 & j=3  \\ \hline
	i=0 & 0.36 & 0.30 & 0.15 & 0.13  \\ \hline
	i=1 & 0.35 & 0.36 & 0.26 & 0.27 \\ \hline
    i=2 & 0.21 & 0.24 & 0.39 & 0.19 \\ \hline
    i=3 & 0.08 & 0.09 & 0.20 & 0.42 \\ \hline
\end{array}\]
\label{tabmismat1}
\end{table}
\begin{table}[!ht] \caption{$\mathbb{P}[X=i | X^*=j]$ (controls)}
\[\begin{array}{ | l | l | l | l | l | l |}
\hline
    & j=0 & j=1 & j=2 & j=3  \\ \hline
	i=0 & 0.31 & 0.26 & 0.20 & 0.16 \\ \hline
    i=1 & 0.31 & 0.34 & 0.13 & 0.34 \\ \hline
    i=2 & 0.22 & 0.24 & 0.50 & 0.09 \\ \hline
    i=3 & 0.16 & 0.16 & 0.18 & 0.42 \\ \hline
\end{array}\]
\label{tabmismat2}
\end{table}
\begin{table}[!ht] \caption{$\mathbb{P}[X=i | X^*=j]$ (controls and FE)}
\[\begin{array}{ | l | l | l | l | l | l |}
\hline
     & j=0 & j=1 & j=2 & j=3  \\ \hline
	i=0 & 0.27 & 0.25 & 0.23 & 0.24 \\ \hline
    i=1 & 0.25 & 0.41 & 0.27 & 0.22 \\ \hline
    i=2 & 0.24 & 0.17 & 0.27 & 0.24 \\ \hline
    i=3 & 0.24 & 0.17 & 0.23 & 0.30 \\ \hline
\end{array}\]
\caption*{Misclassification probability matrices}
\label{tabmismat3}
\end{table}

Overall, our results reinforce those of \citet{michalopoulos2013pre} and, if anything, uncover an even stronger relationship between pre-colonial centralization and current development. Not only are the point estimates of the coefficients larger, but their statistical significance also remains very high. Our proposed approach thus seems to provide a practical and feasible way to address measurement error issues at no extra data collection cost in spatial settings. This capability should prove especially useful in the context of noisy historical data and, more broadly, in any noisy data setting where observation pairs can be assigned a quantifiable notion of ``proximity''. This not only includes geographically tagged data, but also more abstract spaces, such as product or consumer characteristics or network data.

\section{Conclusion}

We have shown that the use of spatial data provides a formal and effective
way to correct for the presence of potentially nonclassical covariate
measurement error in general nonlinear model without relying on
distributional assumptions. Using neighboring observations as repeated
measurements requires carefully controlling for the nonuniform spacing
between observations by constructing the joint distribution of all
measurements as a function of the distance between observations, in order to
ensure that the resulting measurement system satisfies the appropriate
conditional independence restrictions needed to establish identification of
the model.

The method's implementation combines a sieve semiparametric maximum
likelihood with a first-step kernel smoothing estimator and
simulation methods. Monte Carlo simulations suggest that this implementation
performs well at typically available sample sizes.

The method's effectiveness is further illustrated by revisiting a well-known
study of the effect of pre-colonial political structure on current economic
development in Africa. Our estimator support the authors' original findings
by showing that their results are robust to allowing for the likely
possibility that political structure is measured with error. Our results
suggest that the studied relationship could even be stronger than previously
thought.

Our approach opens the way to considering much broader classes of repeated
measurements than previously thought possible, as long as a well-defined
notion of proximity between pairs of observations can be defined. Beyond
geographical data, this could be applicable to network data as well as more
abstract spaces of consumer or product characteristics.

\bibliographystyle{aea}
\bibliography{References}

@article{hu2008instrumental,
  title={Instrumental variable treatment of nonclassical measurement error models},
  author={Hu, Yingyao and Schennach, Susanne M},
  journal={Econometrica},
  volume={76},
  number={1},
  pages={195--216},
  year={2008},
  publisher={Wiley Online Library}
}

@book{zhou2017spatial,
  title={Spatial Data Handling in Big Data Era},
  author={Zhou, Chenghu and Su, Fenzhen and Harvey, Francis and Xu, Jun},
  year={2017},
  publisher={Springer}
}

@article{hu2008nonparametric,
        title="Nonparametric identification of dynamic models with unobserved state variables",
        author="Y. Hu and M. Shum",
        journal="Journal of Econometrics",
        pages="32--44",
        volume=171,
        year=2012
}

@article{andrews:semikern,
        title="Nonparametric Kernel Estimation for Semiparametric Models",
        author="D. W. K. Andrews",
        journal="Econometric Theory",
        volume=11,
        pages="560--596",
        year="1995"
}

@incollection{Newey:HB,
  author="W. Newey and D. McFadden",
  title="Large Sample Estimation and Hypothesis Testing",
  booktitle="Handbook of Econometrics",
  volume="IV",
  editor="R. F. Engel and D. L. McFadden",
  publisher="Elsevier Science",
  year=1994
}

@article{newey:semiparam,
        title="The Asymptotic Variance of Semiparametric Estimators",
        author="W. Newey",
        journal="Econometrica",
        volume=62,
        pages="1349--1382",
        year=1994
}

@article{cunha2010estimating,
  title={Estimating the technology of cognitive and noncognitive skill formation},
  author={Cunha, Flavio and Heckman, James J and Schennach, Susanne M},
  journal={Econometrica},
  volume={78},
  number={3},
  pages={883--931},
  year={2010},
  publisher={Wiley Online Library}
}

@article{michalopoulos2013pre,
  title={Pre-colonial ethnic institutions and contemporary African development},
  author={Michalopoulos, Stelios and Papaioannou, Elias},
  journal={Econometrica},
  volume={81},
  number={1},
  pages={113--152},
  year={2013},
  publisher={Wiley Online Library}
}

@book{murdock1967ethnographic,
  title={Ethnographic atlas},
  author={Murdock, George Peter},
  publisher="University of Pittsburgh Press",
  year={1969}
}

@article{pinkse2010future,
  title={The future of spatial econometrics},
  author={Pinkse, Joris and Slade, Margaret E},
  journal={Journal of Regional Science},
  volume={50},
  number={1},
  pages={103--117},
  year={2010},
  publisher={Wiley Online Library}
}

@article{redding2017quantitative,
  title={Quantitative spatial economics},
  author={Redding, Stephen J and Rossi-Hansberg, Esteban},
  journal={Annual Review of Economics},
  volume={9},
  pages={21--58},
  year={2017},
  publisher={Annual Reviews}
}

@article{hu2008identification,
  title={Identification and estimation of nonlinear models with misclassification error using instrumental variables: A general solution},
  author={Hu, Yingyao},
  journal={Journal of Econometrics},
  volume={144},
  number={1},
  pages={27--61},
  year={2008},
  publisher={Elsevier}
}

@article{griliches:panel,
	title="Errors in Variables in panel data",
	author="Z. Griliches and J. A. Hausman",
	journal="Journal of Econometrics",
	volume=31,
	pages="93--118",
	year="1986"
}

@article{newey:npiv,
        author="W. K. Newey and J. L. Powell",
        title="Instrumental Variable Estimation of Nonparametric Models",
        journal="Econometrica",
        volume=71,
        pages="1565--1578",
        year="2003"
}

@article{hallhorowitz:npiv,
        author="P. Hall and J. L. Horowitz",
        title="Nonparametric Methods for Inference in the Presence of Instrumental Variables",
        journal="Annals of Statistics",
        volume=33,
        pages="2904--2929",
        year=2005
}

@article{schennach:annrev,
        author="S. M. Schennach",
        title="Recent Advances in the Measurement Error Literature",
        volume="8",
        journal="Annual Reviews of Economics",
        pages="341--377",
        year=2016
}

@article{kruskal:3way,
        author="J. B. Kruskal",
        year=1977,
        title="Three-Way Arrays: Rank and Uniqueness of Trilinear Decompositions, with Applications to Arithmetic Complexity and Statistics",
        journal="Linear Algebra and its Applications",
        volume=18,
        pages="95--138"
}

@incollection{paula:econet,
        author="A. de Paula",
        title="Econometrics of Network Models",
        editor="B. Honoré and A. Pakes and M. Piazzesi and L. Samuelson",
        publisher="Cambridge University Press",
        pages="268--323",
        year=2017,
        booktitle="Advances in Economics and Econometrics: Eleventh World Congress",
        chapter="8"
}

@techreport{hu:finmix,
  title="Global estimation of finite mixture and misclassication models with an application to multiple equilibria",
  author="Y. Hu and R. Xiao",
  institution="CeMMAP",
  number="CWP32/18",
  year=2018
}

@article{Shen:sieve,
  author = {X. Shen},
  year = 1997,
  title = {On Methods of Sieves and Penalization},
  journal = {Annals of Statistics},
  volume = 25,
  pages = {2555-2591}
}

@article{bolthausen1982central,
  title={On the central limit theorem for stationary mixing random fields},
  author={Bolthausen, Erwin},
  journal={The Annals of Probability},
  pages={1047--1050},
  year={1982},
  publisher={JSTOR}
}

@article{lahiri2003central,
  title={Central limit theorems for weighted sums of a spatial process under a class of stochastic and fixed designs},
  author={Lahiri, SN},
  journal={Sankhy{\=a}: The Indian Journal of Statistics},
  pages={356--388},
  year={2003},
  publisher={JSTOR}
}

@article{nordman2007optimal,
  title={Optimal block size for variance estimation by a spatial block bootstrap method},
  author={Nordman, Daniel J and Lahiri, Soumendra N and Fridley, Brooke L},
  journal={Sankhy{\=a}: The Indian Journal of Statistics},
  pages={468--493},
  year={2007},
  publisher={JSTOR}
}

@article{carbon1997kernel,
  title={Kernel density estimation for random fields (density estimation for random fields)},
  author={Carbon, Michel and Tran, Lanh Tat and Wu, Berlin},
  journal={Statistics \& Probability Letters},
  volume={36},
  number={2},
  pages={115--125},
  year={1997},
  publisher={Elsevier}
}

@article{jenish2009central,
  title={Central limit theorems and uniform laws of large numbers for arrays of random fields},
  author={Jenish, Nazgul and Prucha, Ingmar R},
  journal={Journal of econometrics},
  volume={150},
  number={1},
  pages={86--98},
  year={2009},
  publisher={Elsevier}
}

@article{jenish2012spatial,
  title={On spatial processes and asymptotic inference under near-epoch dependence},
  author={Jenish, Nazgul and Prucha, Ingmar R},
  journal={Journal of econometrics},
  volume={170},
  number={1},
  pages={178--190},
  year={2012},
  publisher={Elsevier}
}

@article{hall1995blocking,
  title={On blocking rules for the bootstrap with dependent data},
  author={Hall, Peter and Horowitz, Joel L and Jing, Bing-Yi},
  journal={Biometrika},
  volume={82},
  number={3},
  pages={561--574},
  year={1995},
  publisher={Oxford University Press}
}

@article{henderson2012measuring,
  title={Measuring economic growth from outer space},
  author={Henderson, J Vernon and Storeygard, Adam and Weil, David N},
  journal={The American Economic Review},
  volume={102},
  number={2},
  pages={994--1028},
  year={2012},
  publisher={American Economic Association}
}

@article{elvidge1997relation,
  title={Relation between satellite observed visible-near infrared emissions, population, economic activity and electric power consumption},
  author={Elvidge, Christopher D and Baugh, Kimberley E and Kihn, Eric A and Kroehl, Herbert W and Davis, Ethan R and Davis, Chris W},
  journal={International Journal of Remote Sensing},
  volume={18},
  number={6},
  pages={1373--1379},
  year={1997},
  publisher={Taylor \& Francis}
}

@article{doll2006mapping,
  title={Mapping regional economic activity from night-time light satellite imagery},
  author={Doll, Christopher NH and Muller, Jan-Peter and Morley, Jeremy G},
  journal={Ecological Economics},
  volume={57},
  number={1},
  pages={75--92},
  year={2006},
  publisher={Elsevier}
}

@incollection{schennach:hb,
        author="S. M. Schennach",
        title="Mismeasured and unobserved variables",
        booktitle="Handbook of Econometrics",
        volume="7A, invited, under (minor) revision",
        editors="S. Durlauf, L. P. Hansen, J. Heckman and R. L. Matzkin",
        publisher="Elsevier Science",
        year=2018
}

@book{jolliffe:pcamarket,
        title="Principal component analysis",
        author="I. T. Jolliffe",
        publisher="Spinger-Verlag",
        year="1986",
        address="New York"
}

@techreport{schennach:nlfact,
        author="F. Gunsilius and S. M. Schennach",
        title="Independent Principal Component Analysis",
        type="{W}orking {P}aper",
        year=2019,
        institution="Cemmap",
        number="CWP46/19"
}

@article{szipiro:misalligned,
    author = {Szpiro, Adam A. and Sheppard, Lianne and Lumley, Thomas},
    title = "{Efficient measurement error correction with spatially misaligned data}",
    journal = {Biostatistics},
    volume = {12},
    pages = {610-623},
    year = {2011}
}

@article{carroll:spatialsimex,
    author = {Alexeeff, Stacey E. and Carroll, Raymond J. and Coull, Brent},
    title = "{Spatial measurement error and correction by spatial SIMEX in linear regression models when using predicted air pollution exposures}",
    journal = {Biostatistics},
    volume = {17},
    pages = {377-389},
    year = {2016}
}

@article{prusha:spatialiv,
    author="H. H. Kelejian and I. R. Prucha",
    year=1998,
    title="A Generalized Spatial Two-Stage Least Squares Procedure for Estimating a Spatial Autoregressive Model with Autoregressive Disturbances",
    journal="Journal of Real Estate Finance and Economics",
    pages="99-121",
    volume=17
}

@mastersthesis{krige:thesis,
    author="D. G. Krige",
    title="A statistical approach to some mine valuations and allied problems at the Witwatersrand",
    school="University of Witwatersrand",
    year=1951
}

@incollection{chiles:kriging,
  author="J. P. Chil{\`{e}}s and N. Desassis",
  year=2018,
  title="Fifty Years of Kriging",
  editor="B. Daya Sagar and Q. Cheng and F. Agterberg",
  booktitle="Handbook of Mathematical Geosciences",
  publisher="Springer"
}

@book{cressie:kriking,
    author="N. Cressie",
    year=1993,
    title="Statistics for Spatial Data",
    publisher="Wiley Interscience",
    address="New York"
}

@article{Amem:IV,
   author="Amemiya, Y.",
   title="Instrumental Variable Estimator for the Nonlinear Errors-in-Variables Model",
   journal="Journal of Econometrics",
   volume=28,
   year=1985,
   pages="273--289"
}

@article{conley1999gmm,
  title={GMM estimation with cross sectional dependence},
  author={Conley, Timothy G},
  journal={Journal of econometrics},
  volume={92},
  number={1},
  pages={1--45},
  year={1999},
  publisher={Elsevier}
}

@article{sun2018estimation,
  title={Estimation and inference in functional-coefficient spatial autoregressive panel data models with fixed effects},
  author={Sun, Yiguo and Malikov, Emir},
  journal={Journal of Econometrics},
  volume={203},
  number={2},
  pages={359--378},
  year={2018},
  publisher={Elsevier}
}

@article{hallin2004kernel,
  title={Kernel density estimation for spatial processes: the L1 theory},
  author={Hallin, Marc and Lu, Zudi and Tran, Lanh T},
  journal={Journal of Multivariate Analysis},
  volume={88},
  number={1},
  pages={61--75},
  year={2004},
  publisher={Elsevier}
}

@article{sun2016functional,
  title={Functional-coefficient spatial autoregressive models with nonparametric spatial weights},
  author={Sun, Yiguo},
  journal={Journal of econometrics},
  volume={195},
  number={1},
  pages={134--153},
  year={2016},
  publisher={Elsevier}
}

@article{bramoulle:peernet,
  title="Identification of peer effects through social networks",
  author="Y. Bramoull{\'e} and H. Djebbari and B. Fortin",
  journal="Journal of Econometrics",
  volume=150,
  year=2009,
  pages="41--55"
}

@incollection{chen:hb,
	author="X. Chen",
	title="Large Sample Sieve Estimation of Semi-Nonparametric Models",
	booktitle="Handbook of Econometrics",
	volume="Vol. 6",
	publisher="Elsevier Science",
	year=2005
}

@article{carlstein:blockboot,
       author="E. Carlstein",
       title="The use of subseries methods for estimating the variance of a general statistic from a stationary time series",
       journal="Annals of Statistics",
       volume=14,
       pages="1171--1179",
       year=1986
}

@article{schennach:berkson,
	title="Regressions with {B}erkson errors in covariates --- A nonparametric approach",
	author="S. M. Schennach",
	journal="Annals of Statistics",
	volume=41,
	pages="1642--1668",
	year=2013
}

@article{laan:mlecv,
	author="M. J. van der Laan and S. Dudoit and S. Keles",
	title="Asymptotic optimality of likelihood-based cross-validation",
	journal="Statistical Applications in Genetics and Molecular Biology",
	volume=3,
	pages="4",
	year=2004
}

@book{vaart:asymp,
   author="A. W. van der Vaart",
   title="Asymptotic Statistics",
   publisher="Cambridge University Press",
   year=1998
}

@article{bickel:boot,
	author="P. J. Bickel and D. A. Freedman",
	title="Some asymptotic theory for the bootstrap",
	journal="Annals of Statistics",
	year=1981,
	volume=9,
	pages="1196--1217"
}

@article{botev2010kernel,
  title={Kernel density estimation via diffusion},
  author={Botev, Zdravko I and Grotowski, Joseph F and Kroese, Dirk P},
  journal={The annals of Statistics},
  volume={38},
  number={5},
  pages={2916--2957},
  year={2010},
  publisher={Institute of Mathematical Statistics}
}

@article{gunsilius2023independent,
  title={Independent nonlinear component analysis},
  author={Gunsilius, Florian and Schennach, Susanne},
  journal={Journal of the American Statistical Association},
  volume={118},
  number={542},
  pages={1305--1318},
  year={2023},
  publisher={Taylor \& Francis}
}
\appendix

\section{Proofs}

\label{appproof}

\begin{proof}[Theorem \protect\ref{thmmain}]
We handle the case of discrete and continuous $\mu_{X}$ separately. 
%We let $=^d$ denote equality in distribution.

For the continuous case, we show that assumptions 1 to 5 in %
\citet{hu2008instrumental} are satisfied in our framework. Identification
then follows from their Theorem 1.

First, assumption \ref{asscondindep} implies their assumption 2, both (i)
and (ii). For (i), we observe that 
\begin{eqnarray}
f_{Y( s) \mid X( s) ,X( s+\Delta s) ,X^{\ast }( s) }( y\mid x,z,x^{\ast })
&=& \frac{f_{Y( s) ,X( s) ,X( s+\Delta s) \mid X^{\ast }( s) }( y,x,z\mid
x^{\ast }) }{f_{X( s) ,X( s+\Delta s) \mid X^{\ast }( s) }( x,z\mid x^{\ast
}) }  \notag \\
&=& \frac{f_{Y( s) \mid X^{\ast }( s) }( y\mid x^{\ast }) f_{X( s) \mid
X^{\ast }( s) }( x\mid x^{\ast }) f_{X( s+\Delta s) \mid X^{\ast }( s) }(
z\mid x^{\ast }) }{f_{X( s) \mid X^{\ast }( s) }( x\mid x^{\ast }) f_{X(
s+\Delta s) \mid X^{\ast }( s) }( z\mid x^{\ast }) }  \notag \\
&=& f_{Y( s) \mid X^{\ast }( s) }( y\mid x^{\ast })  \label{eqcii}
\end{eqnarray}
where assumption \ref{asscondindep} was used to factor the densities as
product of conditional marginals. Next, to establish (ii), we similarly have 
\begin{eqnarray}
f_{X( s) \mid X( s+\Delta s) ,X^{\ast }( s) }( x\mid z,x^{\ast }) &=& \frac{%
f_{X( s) ,X( s+\Delta s) \mid X^{\ast }( s) }( x,z\mid x^{\ast }) }{f_{X(
s+\Delta s) \mid X^{\ast }( s) }( z\mid x^{\ast }) }  \notag \\
&=& \frac{f_{X( s) \mid X^{\ast }( s) }( x\mid x^{\ast }) f_{X( s+\Delta s)
\mid X^{\ast }( s) }( z\mid x^{\ast }) }{f_{X( s+\Delta s) \mid X^{\ast }(
s) }( z\mid x^{\ast }) }  \notag \\
&=& f_{X( s) \mid X^{\ast }( s) }( x\mid x^{\ast }) .  \label{eqciii}
\end{eqnarray}

Assumptions \ref{densities}, \ref{injectivity}, and \ref{variation} are
direct counterparts of assumptions 1, 3, and 4 in \citet{hu2008instrumental}
adapted to our spatial setup. Finally, the existence of $M_x$ in assumption %
\ref{centering} establishes their assumption 5.

Hence, by Theorem 1 in \citet{hu2008instrumental}, the knowledge of $%
f_{Y(s), X(s), X(s+\Delta s)}(y,x,z)$ identifies $f_{Y(s)|X^{*}(s)}$, $%
f_{X(s)|X^{*}(s)}$, $f_{X(s+\Delta s)|X^{*}(s)}$, and $f_{X^{*}(s)}$.

For the discrete case, we first show that our assumptions imply the
assumptions 1, 2, 2.1, 2.2 of \citet{hu2008identification}. Note that their
assumptions explicitly include possible conditioning on a covariate $w$,
while our notation leaves such conditioning implicit, for simplicity.

Our assumption \ref{asscondindep} implies their assumption 1 and 2, by the
same reasoning that lead to Equations (\ref{eqcii}) and (\ref{eqciii})
above. Next, our assumption \ref{injectivity} reduces to their assumptions
2.1 and 2.2 in the discrete case, since the integral operators reduce to
matrix multiplications when $\mu_X$ is discrete: $[L_{B|A}h] (b) = \int
f_{B|A}(b|a) h(a) d\mu_X(a) = \sum_a F_{B|A}(b|a) h(a) \mu(\{a\})$.

Finally, although none of our assumptions imply one of their set of
alternative assumptions 2.3 through 2.7, these assumptions are only needed
to secure the proper ordering of the possible values of the latent discrete
variable $X^{*}$. Any re-ordering of it implies a re-ordering of the column
of the matrix $f_{X(s)|X^{*}(s)}(x|x^{*})$. However, any ordering other than
the correct one would lead to a violation of our assumption \ref{centering}: 
$M_x[f_{X(s)|X^{*}(s)}( \cdot | x^{*})] = x^{*}$. Hence our assumption \ref%
{centering} has the same effect as their set of alternative assumptions 2.3
through 2.7. (Note that in the special case where $M_x$ is the mode
functional, our assumption \ref{centering} regarding $X(s)|X^{*}(s)$ is the
same as their assumption 2.7.)

From the above consideration, we can invoke their Theorem 1 to establish
identification of our model in the discrete case.
\end{proof}

\begin{proof}[Proof of Theorem \protect\ref{thasslin}]
We take the following convention to ensure that the $Z_{i}$ vary smoothly as 
$f$ is changed in the expression $\mathcal{\hat{L}}\left( \theta ,f\right) $
for $f\not=\hat{f}$. Letting $F^{-1}\left( \cdot |x,y,w\right) $ denotes the
inverse of the cdf of $Z$ given $X$, $Y$ and $W$ with respect to the first
argument, we set $Z_{i}=\hat{F}_{Z|X,Y,W}^{-1}\left(
U_{i}|X_{i},Y_{i},W_{i}\right) $ (in the unidimensional case\footnote{%
In the the multivariate $Z_{i}$ case, one proceeds iteratively, starting
with $Z_{i,1}=F_{Z_{1}|X,Y}^{-1}\left( U_{i,1}|X_{i},Y_{i}\right) $ and
continuing with $Z_{i,k}=F_{Z_{k}|Z_{1},\ldots ,Z_{k-1},X,Y}^{-1}\left(
U_{i,k}|Z_{i,1},\ldots ,Z_{i,k-1},X_{i},Y_{i}\right) $ for $k=2,\ldots ,\dim
Z_{i}$ and with all $U_{i,k}$ mutually independent.}) where $U_{i}$ is drawn
from a uniform and the $U_{i}$ are kept fixed as $f$ varies. This is purely
a device of proof and a harmless convention because $\mathcal{\hat{L}}\left(
\theta ,f\right) $ is only evaluated at $f=\hat{f}$ in the estimator.
However, the structure of the proof (which uses constructs involving $%
\mathcal{\hat{L}}\left( \theta ,f\right) $ for $f\not=\hat{f}$) is
considerably simplified with this convention.

We first show consistency. Conclusions (i) and (ii) or Lemma \ref{lemwasass4} and \ref%
{asssupp}(i) imply that $\left\Vert \hat{f}-f_{0}\right\Vert \overset{p}{%
\longrightarrow }0$. To show that $\hat{\theta}\overset{p}{\longrightarrow }%
\theta $, we observe that, by the triangular inequality,%
\begin{equation*}
\left\vert \mathcal{\hat{L}}\left( \theta ,\hat{f}\right) -\mathcal{L}\left(
\theta ,f_{0}\right) \right\vert \leq \left\vert \mathcal{\hat{L}}\left(
\theta ,\hat{f}\right) -\mathcal{L}\left( \theta ,\hat{f}\right) \right\vert
+\left\vert \mathcal{L}\left( \theta ,\hat{f}\right) -\mathcal{L}\left(
\theta ,f_{0}\right) \right\vert .
\end{equation*}%
The first term satisfies $\left\vert \mathcal{\hat{L}}\left( \theta ,\hat{f}%
\right) -\mathcal{L}\left( \theta ,\hat{f}\right) \right\vert \overset{p}{%
\longrightarrow }0$ by Assumption \ref{asscons}(ii) and the fact that
eventually $\hat{f}\in \mathcal{F}$ since $\hat{f}\overset{p}{%
\longrightarrow }f_{0}$. The second term is also such that $\left\vert 
\mathcal{L}\left( \theta ,\hat{f}\right) -\mathcal{L}\left( \theta
,f_{0}\right) \right\vert \overset{p}{\longrightarrow }0$ by Assumption \ref%
{asscons}(iii) and $\hat{f}\overset{p}{\longrightarrow }f_{0}$. Since $%
\mathcal{\hat{L}}\left( \theta ,\hat{f}\right) $ converges uniformly to a
function that is uniquely maximized at $\theta _{0}$ (by Assumption \ref%
{asscons}(i)), it follows that $\hat{\theta}=\arg \max_{\theta \in \Theta }%
\mathcal{\hat{L}}\left( \theta ,\hat{f}\right) \overset{p}{\longrightarrow }%
\arg \max_{\theta \in \Theta }\mathcal{L}\left( \theta ,f_{0}\right) =\theta
_{0}$, by Theorem 2.1 in \cite{Newey:HB}.

By a standard expansion of the first order conditions $\nabla \mathcal{\hat{L%
}}\left( \hat{\theta},\hat{f}\right) =0$ around the true value $\theta
=\theta _{0}$, we have:%
\begin{equation*}
\nabla \mathcal{\hat{L}}\left( \theta _{0},\hat{f}\right) +\nabla \nabla
^{\prime }\mathcal{\hat{L}}\left( \bar{\theta},\hat{f}\right) \left( \hat{%
\theta}-\theta _{0}\right) =0
\end{equation*}%
where $\bar{\theta}$ is mean value between $\theta _{0}$ and $\hat{\theta}$.
Rearranging, we have%
\begin{eqnarray*}
&~&n^{1/2}\left( \hat{\theta}-\theta _{0}\right)  \\
&=&-n^{1/2}\left( \nabla \nabla ^{\prime }\mathcal{\hat{L}}\left( \bar{\theta%
},\hat{f}\right) \right) ^{-1}\nabla \mathcal{\hat{L}}\left( \theta _{0},%
\hat{f}\right)  \\
&=&-n^{1/2}\left( \nabla \nabla ^{\prime }\mathcal{\hat{L}}\left( \bar{\theta%
},\hat{f}\right) \right) ^{-1}\left( \nabla \mathcal{\hat{L}}\left( \theta
_{0},\hat{f}\right) -\nabla \mathcal{L}\left( \theta _{0},\hat{f}\right)
+\nabla \mathcal{L}\left( \theta _{0},\hat{f}\right) -\nabla \mathcal{L}%
\left( \theta _{0},f_{0}\right) \right)  \\
&=&\hat{\Psi}_{\text{MLE}}+\hat{\Psi}_{\text{kernel}}+\hat{R}_{1}
\end{eqnarray*}%
where we have inserted $-\nabla \mathcal{L}\left( \theta _{0},\hat{f}\right)
+\nabla \mathcal{L}\left( \theta _{0},\hat{f}\right) =0$ and $\nabla 
\mathcal{L}\left( \theta _{0},f_{0}\right) =0$ (by construction) and where%
\begin{eqnarray*}
\hat{\Psi}_{\text{MLE}} &=&-n^{1/2}\hat{H}^{-1}\left( \nabla \mathcal{\hat{L}%
}\left( \theta _{0},f_{0}\right) -\nabla \mathcal{L}\left( \theta
_{0},f_{0}\right) \right)  \\
\hat{\Psi}_{\text{kernel}} &=&-n^{1/2}\hat{H}^{-1}\left( \nabla \mathcal{L}%
\left( \theta _{0},\hat{f}\right) -\nabla \mathcal{L}\left( \theta
_{0},f_{0}\right) \right)  \\
\hat{R}_{1} &=&-n^{1/2}\hat{H}^{-1}\left( \left( \nabla \mathcal{\hat{L}}%
\left( \theta _{0},\hat{f}\right) -\nabla \mathcal{L}\left( \theta _{0},\hat{%
f}\right) \right) -\left( \nabla \mathcal{\hat{L}}\left( \theta
_{0},f_{0}\right) -\nabla \mathcal{L}\left( \theta _{0},f_{0}\right) \right)
\right)  \\
\hat{H} &=&\nabla \nabla ^{\prime }\mathcal{\hat{L}}\left( \bar{\theta},\hat{%
f}\right) .
\end{eqnarray*}%
We first show that $\hat{H}\overset{p}{\longrightarrow }H\equiv \nabla
\nabla ^{\prime }\mathcal{L}\left( \theta _{0},f_{0}\right) $ as follows:%
\begin{equation*}
\hat{H}-H=\left( \nabla \nabla ^{\prime }\mathcal{\hat{L}}\left( \bar{\theta}%
,f_{0}\right) -\nabla \nabla ^{\prime }\mathcal{L}\left( \theta
_{0},f_{0}\right) \right) +\left( \nabla \nabla ^{\prime }\mathcal{\hat{L}}%
\left( \bar{\theta},\hat{f}\right) -\nabla \nabla ^{\prime }\mathcal{\hat{L}}%
\left( \bar{\theta},f_{0}\right) \right) 
\end{equation*}%
where the first term is such that $\left( \nabla \nabla ^{\prime }\mathcal{%
\hat{L}}\left( \bar{\theta},f_{0}\right) -\nabla \nabla ^{\prime }\mathcal{L}%
\left( \theta _{0},f_{0}\right) \right) \overset{p}{\longrightarrow }0$ from
Assumption \ref{asssieve}(i) and (iv), while the second term can be written
as:%
\begin{eqnarray*}
\nabla \nabla ^{\prime }\mathcal{\hat{L}}\left( \bar{\theta},\hat{f}\right)
-\nabla \nabla ^{\prime }\mathcal{\hat{L}}\left( \bar{\theta},f_{0}\right) 
&=&\left( \nabla \nabla ^{\prime }\mathcal{\hat{L}}\left( \bar{\theta},\hat{f%
}\right) -\nabla \nabla ^{\prime }\mathcal{L}\left( \bar{\theta},\hat{f}%
\right) \right)  \\
&&-\left( \nabla \nabla ^{\prime }\mathcal{\hat{L}}\left( \bar{\theta}%
,f_{0}\right) -\nabla \nabla ^{\prime }\mathcal{L}\left( \bar{\theta}%
,f_{0}\right) \right)  \\
&&-\left( \nabla \nabla ^{\prime }\mathcal{L}\left( \bar{\theta}%
,f_{0}\right) -\nabla \nabla ^{\prime }\mathcal{L}\left( \bar{\theta},\hat{f}%
\right) \right) .
\end{eqnarray*}%
The two first term converge in probability to zero by Assumption \ref%
{asssieve}(i) and the fact that eventually $\hat{f}\in \mathcal{F}$, by
conclusions (i) and (ii) of Lemma \ref{lemwasass4} and \ref{asssupp}(i). The last term
converges in probability to $0$ since, by Assumption \ref{asssieve}(iii), 
\begin{equation*}
\plim_{n\longrightarrow \infty }\nabla \nabla ^{\prime }\mathcal{L}\left(
\theta ,\hat{f}\right) =\nabla \nabla ^{\prime }\mathcal{L}\left( \theta ,%
\plim_{n\longrightarrow \infty }\hat{f}\right) =\nabla \nabla ^{\prime }%
\mathcal{L}\left( \theta ,f_{0}\right) \text{ uniformly for }\theta \in
\Theta \text{.}
\end{equation*}%
It follows that $\hat{H}\overset{p}{\longrightarrow }H$. By assumption \ref%
{asssieve}(ii), we also have $\hat{H}^{-1}\overset{p}{\longrightarrow }H^{-1}
$, so that $\hat{\Psi}_{\text{MLE}}-\Psi _{\text{MLE}}\overset{p}{%
\longrightarrow }0$, $\hat{\Psi}_{\text{kernel}}-\tilde{\Psi}_{\text{kernel}}%
\overset{p}{\longrightarrow }0$ and $\hat{R}_{1}-R_{1}\overset{p}{%
\longrightarrow }0$ for 
\begin{eqnarray*}
\Psi _{\text{MLE}} &=&-n^{1/2}H^{-1}\left( \nabla \mathcal{\hat{L}}\left(
\theta _{0},f_{0}\right) -\nabla \mathcal{L}\left( \theta _{0},f_{0}\right)
\right) =-n^{-1/2}H^{-1}\sum_{i=1}^{n}\psi _{MLE}\left(
Y_{i},X_{i},Z_{i},W_{i}\right)  \\
\tilde{\Psi}_{\text{kernel}} &=&-n^{1/2}H^{-1}\left( \nabla \mathcal{L}%
\left( \theta _{0},\hat{f}\right) -\nabla \mathcal{L}\left( \theta
_{0},f_{0}\right) \right)  \\
R_{1} &=&-n^{1/2}H^{-1}\left( \left( \nabla \mathcal{\hat{L}}\left( \theta
_{0},\hat{f}\right) -\nabla \mathcal{L}\left( \theta _{0},\hat{f}\right)
\right) -\left( \nabla \mathcal{\hat{L}}\left( \theta _{0},f_{0}\right)
-\nabla \mathcal{L}\left( \theta _{0},f_{0}\right) \right) \right) .
\end{eqnarray*}%
where $\psi _{\text{MLE}}\left( y,x,z,w\right) =\nabla \ln L\left(
y,x,z,w;\theta _{0},\omega \left( \theta _{0}\right) \right) $ is the usual
influence function of a sieve MLE estimator of $\theta _{0}$, while 
\begin{eqnarray*}
\tilde{\Psi}_{\text{kernel}} &=& -n^{1/2}H^{-1}\iiiint \left( \hat{f}\left( z|y,x,w;\hat{\kappa%
}\right) -f\left( z|y,x,w;\kappa _{0}\right) \right) f_{YXW}\left(
y,x,w\right)\times\\
&& \nabla \ln L\left( y,x,z,w;\theta _{0},\omega \left( \theta
_{0}\right) \right) dydxdzdw \\
&=&\tilde{\Psi}_{\text{kernel}}^{1}+\tilde{\Psi}_{\text{kernel}}^{2}+R_{2}
\end{eqnarray*}%
where $f$ denotes $f_{Z|YXW}$ (as in the definition of our estimator) and
where%
\begin{eqnarray*}
\tilde{\Psi}_{\text{kernel}}^{1} &=&-n^{1/2}H^{-1}\iiiint \left( 
\hat{f}\left( z|y,x,w;\kappa _{0}\right) -f\left( z|y,x,w;\kappa _{0}\right)
\right) f_{YXW}\left( y,x,w\right) \times\\
&&\nabla \ln L\left( y,x,z,w;\theta
_{0},\omega \left( \theta _{0}\right) \right) dydxdzdw \\
\tilde{\Psi}_{\text{kernel}}^{2} &=&-n^{1/2}H^{-1}\iiiint \left(
f\left( z|y,x,w;\hat{\kappa}\right) -f\left( z|y,x,w;\kappa _{0}\right)
\right) f_{YXW}\left( y,x,w\right) \times\\
&& \nabla \ln L\left( y,x,z,w;\theta
_{0},\omega \left( \theta _{0}\right) \right) dydxdzdw \\
R_{2} &=&-n^{1/2}H^{-1}\iiiint \left( \left( \hat{f}\left(
z|y,x,w;\hat{\kappa}\right) -\hat{f}\left( z|y,x,w;\kappa _{0}\right)
\right) -\left( f\left( z|y,x,w;\hat{\kappa}\right) -f\left( z|y,x,w;\kappa
_{0}\right) \right) \right)  \\
&&\times f_{YXW}\left( y,x,w\right) \nabla \ln L\left( y,x,z,w;\theta
_{0},\omega \left( \theta _{0}\right) \right) dydxdzdw.
\end{eqnarray*}%
We can re-write these terms in alternative ways (making the dependence on $%
\kappa _{0}$ implicit when not central to the argument):%
\begin{eqnarray*}
\tilde{\Psi}_{\text{kernel}}^{1}&=&-n^{1/2}H^{-1}\iiiint \left( 
\frac{\hat{f}_{ZYX|W}\left( z,y,x|w\right) }{\hat{f}_{YX|W}\left(
y,x|w\right) }-\frac{f_{ZYX|W}\left( z,y,x|w\right) }{f_{YX|W}\left(
y,x|w\right) }\right) f_{YXW}\left( y,x,w\right)\times\\
&&\nabla \ln L\left(y,x,z,w;\theta _{0},\omega \left( \theta _{0}\right) \right) dydxdzdw
\end{eqnarray*}%
and $\tilde{\Psi}_{\text{kernel}}^{1}$ can be further linearized by using
the fact that:%
\begin{eqnarray*}
\left( \frac{\hat{a}}{\hat{b}}-\frac{a}{b}\right)  &=&\left( \frac{\hat{a}-a%
}{b}-\frac{a}{b}\frac{\hat{b}-b}{b}\right) +\left( 1+\frac{\hat{b}-b}{b}%
\right) ^{-1}\left( \frac{a}{b}\left( \frac{\hat{b}-b}{b}\right) ^{2}-\frac{%
\left( \hat{b}-b\right) }{b}\frac{\left( \hat{a}-a\right) }{b}\right)  \\
&=&\left( \frac{\hat{a}-a}{b}-\frac{a}{b}\frac{\hat{b}-b}{b}\right)
+o_{p}\left( n^{-1/2}\right) 
\end{eqnarray*}%
if $\left\Vert \hat{a}-a\right\Vert =o_{p}\left( n^{-1/4}\right) $, $%
\left\Vert \hat{b}-b\right\Vert =o_{p}\left( n^{-1/4}\right) $ and $b\geq
\varepsilon >0$. Setting $a=f_{ZYX|W}\left( z,y,x|w\right) $, $\hat{a}=\hat{f}%
_{ZYX|W}\left( z,y,x|w\right) $, $b=f_{YX|W}\left( y,x|w\right) $ and\ $\hat{b}%
=\hat{f}_{YX|W}\left( y,x|w\right) $ and invoking Lemma \ref{lemwasass4} to
establish the required $o_{p}\left( n^{-1/4}\right) $ rates yields:%
\begin{eqnarray*}
\tilde{\Psi}_{\text{kernel}}^{1} &=&-n^{1/2}H^{-1}\int \int \int \int \frac{1%
}{f_{YX|W}\left( y,x|w\right) }\left( \hat{f}_{ZYX|W}\left( z,y,x|w\right)
-f_{ZYX|W}\left( z,y,x|w\right) \right)  \\
&&\times f_{YXW}\left( y,x,w\right) \nabla \ln L\left( y,x,z,w;\theta
_{0},\omega \left( \theta _{0}\right) \right) dydxdzdw \\
&&+n^{1/2}H^{-1}\int \int \int \int \frac{f_{ZYX|W}\left( z,y,x|w\right) }{%
\left( f_{YX|W}\left( y,x|w\right) \right) ^{2}}\left( \hat{f}_{YX|W}\left(
y,x|w\right) -f_{YX|W}\left( y,x|w\right) \right)  \\
&&\times f_{YXW}\left( y,x,w\right) \nabla \ln L\left( y,x,z,w;\theta
_{0},\omega \left( \theta _{0}\right) \right) dydxdzdw+n^{1/2}o_{p}\left(
n^{-1/2}\right)  \\
&=&\tilde{\Psi}_{\text{kernel}}^{11}+\tilde{\Psi}_{\text{kernel}%
}^{12}+o_{p}\left( 1\right) 
\end{eqnarray*}%
where%
\begin{eqnarray*}
\tilde{\Psi}_{\text{kernel}}^{11} &=&-n^{1/2}H^{-1}\int \int \int \int
\left( \hat{f}_{ZYX|W}\left( z,y,x|w\right) -f_{ZYX|W}\left( z,y,x|w\right)
\right) f_{W}\left( w\right) \times  \\
&&\nabla \ln L\left( y,x,z,w;\theta _{0},\omega \left( \theta _{0}\right)
\right) dydxdzdw \\
\tilde{\Psi}_{\text{kernel}}^{12} &=&n^{1/2}H^{-1}\int \int \int \int \left( 
\hat{f}_{YX|W}\left( y,x|w\right) -f_{YX|W}\left( y,x|w\right) \right)
f_{W}\left( w\right) f_{Z|YXW}\left( z|y,x,w\right) \times  \\
&&\nabla \ln L\left( y,x,z,w;\theta _{0},\omega \left( \theta _{0}\right)
\right) dydxdzdw
\end{eqnarray*}%
Next, we exploit the index structure implied by Equation (\ref{eqgenV}) to
write:%
\begin{eqnarray*}
f_{ZYX|W}\left( z,y,x|w\right)  &=&f_{\tilde{Z}\tilde{Y}\tilde{X}}\left(
G^{-1}\left( (y,x,z),w\right) \right) J\left( z,y,x,w\right)  \\
f_{YX|W}\left( y,x|w\right)  &=&\int f_{ZYX|W}\left( z,y,x|w\right) dz
\end{eqnarray*}%
where $G^{-1}$ denotes an inverse with respect to the first (vectorial)
argument and where the dependence on $\kappa _{0}$ is implicit, while $%
J\left( z,y,x,w\right) =\left( \det \nabla _{(y,x,z)}^{\prime }G\left(
(y,x,z),w\right) \right) ^{-1}$ is a Jacobian term. Similar expressions hold
for the corresponding estimated densities. We now make the change of
variable $(\tilde{y},\tilde{x},\tilde{z})=G^{-1}\left( (y,x,z),w\right) $
and use Fubini's Theorem to obtain%
\begin{eqnarray*}
\tilde{\Psi}_{\text{kernel}}^{11} &=&-n^{1/2}H^{-1}\int \int \int \left( 
\hat{f}_{\tilde{Z}\tilde{Y}\tilde{X}}\left( \tilde{z},\tilde{y},\tilde{x}%
\right) -f_{\tilde{Z}\tilde{Y}\tilde{X}}\left( \tilde{z},\tilde{y},\tilde{x}%
\right) \right) \times  \\
&&\int f_{W}\left( w\right) \nabla \ln L\left( G\left( (\tilde{y},\tilde{x},%
\tilde{z}),w\right) ,w;\theta _{0},\omega \left( \theta _{0}\right) \right)
dw~d\tilde{y}d\tilde{x}d\tilde{z} \\
&=&-n^{1/2}H^{-1}\int \int \int \left( \hat{f}_{\tilde{Z}\tilde{Y}\tilde{X}%
}\left( \tilde{z},\tilde{y},\tilde{x}\right) -f_{\tilde{Z}\tilde{Y}\tilde{X}%
}\left( \tilde{z},\tilde{y},\tilde{x}\right) \right) \nabla \ln \tilde{L}%
_{1}\left( \tilde{y},\tilde{x},\tilde{z};\theta _{0}\right) d\tilde{y}d%
\tilde{x}d\tilde{z}
\end{eqnarray*}%
where%
\begin{eqnarray*}
\nabla \ln \tilde{L}_{1}\left( \tilde{y},\tilde{x},\tilde{z};\theta
_{0}\right)  &=&\int f_{W}\left( w\right) \nabla \ln L\left( G\left( (\tilde{%
y},\tilde{x},\tilde{z}),w\right) ,w;\theta _{0},\omega \left( \theta
_{0}\right) \right) dw \\
&=&E\left[ \nabla \ln L\left( G\left( (\tilde{y},\tilde{x},\tilde{z}%
),W\right) ,W;\theta _{0},\omega \left( \theta _{0}\right) \right) \right] .
\end{eqnarray*}%
Similarly,%
\begin{eqnarray*}
\tilde{\Psi}_{\text{kernel}}^{12} &=&n^{1/2}H^{-1}\int \int \int \left[ \int
\left( \hat{f}_{ZYX|W}\left( z,y,x|w\right) -f_{ZYX|W}\left( zy,x|w\right)
\right) dz\right] \times  \\
&&f_{W}\left( w\right) \left[ \int f_{Z|YXW}\left( z|y,x,w\right) \nabla \ln
L\left( y,x,z,w;\theta _{0},\omega \left( \theta _{0}\right) \right) dz%
\right] dwdydx \\
&=&n^{1/2}H^{-1}\int \int \int \int \left( \hat{f}_{YXZ|W}\left(
y,x,z|w\right) -f_{YXZ|W}\left( y,x,z|w\right) \right) \times  \\
&& f_{W}\left( w\right)
\nabla L_{2}\left( y,x,w,\theta _{0}\right) dzdwdydx
\end{eqnarray*}%
where%
\begin{equation*}
\nabla L_{2}\left( y,x,w,\theta _{0}\right) =\int f_{Z|YXW}\left(
z|y,x,w\right) \nabla \ln L\left( y,x,z,w;\theta _{0},\omega \left( \theta
_{0}\right) \right) dz.
\end{equation*}%
Exploiting the index structure, we have:%
\begin{eqnarray*}
\tilde{\Psi}_{\text{kernel}}^{12} &=& n^{1/2}H^{-1}\int \int \int \int \left( 
\hat{f}_{\tilde{Z}\tilde{Y}\tilde{X}}\left( G^{-1}\left( (y,x,z),w\right)
\right) -f_{\tilde{Z}\tilde{Y}\tilde{X}}\left( G^{-1}\left( (y,x,z),w\right)
\right) \right)\times \\
&&J\left( z,y,x,w\right) f_{W}\left( w\right) \nabla
L_{2}\left( y,x,w,\theta _{0}\right) dzdwdydx
\end{eqnarray*}%
and making the change of variable $(\tilde{y},\tilde{x},\tilde{z}%
)=G^{-1}\left( (y,x,z),w\right) $, yields:%
\begin{eqnarray*}
\tilde{\Psi}_{\text{kernel}}^{12} &=&n^{1/2}H^{-1}\int \int \int \int \left( 
\hat{f}_{\tilde{Z}\tilde{Y}\tilde{X}}\left( \tilde{y},\tilde{x},\tilde{z}%
\right) -f_{\tilde{Z}\tilde{Y}\tilde{X}}\left( \tilde{y},\tilde{x},\tilde{z}%
\right) \right)\times \\
&&f_{W}\left( w\right) \nabla L_{2}\left( G_{yx}\left( (\tilde{%
y},\tilde{x},\tilde{z}),w\right) ,w,\theta _{0}\right) dwd\tilde{y}d\tilde{x}%
d\tilde{z} \\
&=&n^{1/2}H^{-1}\int \int \int \left( \hat{f}_{\tilde{Z}\tilde{Y}\tilde{X}%
}\left( \tilde{y},\tilde{x},\tilde{z}\right) -f_{\tilde{Z}\tilde{Y}\tilde{X}%
}\left( \tilde{y},\tilde{x},\tilde{z}\right) \right)\times \\
&&\left[ \int f_{W}\left(
w\right) \nabla L_{2}\left( G_{yx}\left( (\tilde{y},\tilde{x},\tilde{z}%
),w\right) ,w,\theta _{0}\right) dw\right] d\tilde{y}d\tilde{x}d\tilde{z} \\
&=&n^{1/2}H^{-1}\int \int \int \left( \hat{f}_{\tilde{Z}\tilde{Y}\tilde{X}%
}\left( \tilde{y},\tilde{x},\tilde{z}\right) -f_{\tilde{Z}\tilde{Y}\tilde{X}%
}\left( \tilde{y},\tilde{x},\tilde{z}\right) \right) \nabla \ln \tilde{L}%
_{2}\left( \tilde{y},\tilde{x},\tilde{z};\theta _{0}\right) d\tilde{y}d%
\tilde{x}d\tilde{z}
\end{eqnarray*}%
where $G_{yx}\left( (\tilde{y},\tilde{x},\tilde{z}),w\right) $ denotes the $y
$ and $x$ elements of the vector $G\left( (\tilde{y},\tilde{x},\tilde{z}%
),w\right) $ and where%
\begin{eqnarray*}
\nabla \ln \tilde{L}_{2}\left( \tilde{y},\tilde{x},\tilde{z};\theta
_{0}\right)  &=&\int f_{W}\left( w\right) \nabla L_{2}\left( G_{yx}\left( (%
\tilde{y},\tilde{x},\tilde{z}),w\right) ,w,\theta _{0}\right) dw \\
&=&E\left[ \nabla L_{2}\left( G_{yx}\left( (\tilde{y},\tilde{x},\tilde{z}%
),W\right) ,W,\theta _{0}\right) \right] 
\end{eqnarray*}%
Using standard semiparametric correction terms for density estimation (\cite%
{newey:semiparam}) and under the small bias result of Lemma \ref{lemwasass4}%
, $\tilde{\Psi}_{\text{kernel}}^{1}$ can be shown to be asymptotically
equivalent to sample averages (by Lemma \ref{lemnewey} below, under
Assumptions \ref{asskern2} and \ref{asssemipar}): 
\begin{equation*}
\tilde{\Psi}_{\text{kernel}}^{1}=n^{-1/2}\sum_{i=1}^{n}\psi _{\text{kernel}%
}\left( Y_{i},X_{i},Z_{i},W_{i}\right) +o_{p}\left( 1\right) 
\end{equation*}%
where $\psi _{\text{kernel}}$ is given the Theorem statement.

Next, we can re-express the $\tilde{\Psi}_{\text{kernel}}^{2}$ term as:%
\begin{equation*}
\tilde{\Psi}_{\text{kernel}}^{2}=\tilde{\Psi}_{\text{kernel}}^{21}+R_{3}
\end{equation*}%
where%
\begin{eqnarray*}
R_{3}&=&-n^{1/2}H^{-1}\iiiint \left( \nabla _{\kappa }^{\prime
}f\left( z|y,x,w;\dot{\kappa}\right) -\nabla _{\kappa }^{\prime }f\left(
z|y,x,w;\kappa _{0}\right) \right) f_{YXW}\left( y,x,w\right)\times \\
& &\nabla \ln L\left( y,x,z,w;\theta _{0},\omega \left( \theta _{0}\right) \right)
dydxdzdw~\left( \hat{\kappa}-\kappa _{0}\right) 
\end{eqnarray*}%
and
\begin{eqnarray*}
\tilde{\Psi}_{\text{kernel}}^{21} &=&-n^{1/2}H^{-1}\iiiint
\nabla _{\kappa }^{\prime }f\left( z|y,x,w;\kappa _{0}\right) f_{YXW}\left(
y,x,w\right)\times \\
&&\nabla \ln L\left( y,x,z,w;\theta _{0},\omega \left( \theta
_{0}\right) \right) dydxdzdw~\left( \hat{\kappa}-\kappa _{0}\right)  \\
\end{eqnarray*}%
\begin{eqnarray*}
&=&-n^{1/2}H^{-1}\iiiint \nabla _{\kappa }^{\prime }\ln f\left(
z|y,x,w;\kappa _{0}\right) f_{ZYXW}\left( z,y,x,w\right)\times \\
&&\nabla \ln L\left(y,x,z,w;\theta _{0},\omega \left( \theta _{0}\right) \right) dydxdzdw~\left( \hat{\kappa}-\kappa _{0}\right)  \\
&=&-H^{-1}E\left[ \nabla _{\kappa }^{\prime }\ln f\left(Z|Y,X,W;\kappa_{0}\right) 
\nabla \ln L\left( Y,X,Z,W;\theta _{0},\omega \left( \theta
_{0}\right) \right) \right] \times\\
&&n^{-1/2}\sum_{i=1}^{n}\psi _{\kappa }\left(
Y_{i},X_{i},Z_{i},W_{i}\right) +o_{p}\left( 1\right)  \\
&=&n^{-1/2}\sum_{i=1}^{n}\psi _{\text{cov}}\left(
Y_{i},X_{i},Z_{i},W_{i}\right) +o_{p}\left( 1\right) 
\end{eqnarray*}%
where $\left( \hat{\kappa}-\kappa _{0}\right) $ was replaced by its
expression from Assumption \ref{assnuinf} and where $\psi _{\text{cov}}\left( y,x,z,w\right)$ is given in the Theorem statement.

There remains to show that the remainder term $R_{1},R_{2}$ and $R_{3}$ are $%
o_{p}\left( 1\right) $.

For $R_{1}$, we need to show that $n^{1/2}((\nabla \mathcal{\hat{L}}\left(
\theta _{0},f\right) -\nabla \mathcal{L}\left( \theta _{0},f\right)
)-(\nabla \mathcal{\hat{L}}\left( \theta _{0},f_{0}\right) -\nabla \mathcal{L%
}\left( \theta _{0},f_{0}\right) ))$ is stochastically equicontinuous in $f$
at $f=f_{0}$ for all sufficiently large $n$. This standard property follows
from (a) $\nabla \mathcal{L}\left( \theta _{0},f\right) $ being linear in $f$
with bounded prefactor by Assumption \ref{asszhat}(i), (b) $\nabla \mathcal{%
\hat{L}}\left( \theta _{0},f\right) $ being Lipschitz in each of the $Z_{i}$
by Assumption \ref{asszhat}(i) and (c) the $Z_{i}$ being Lipschitz in $f$
(in the sup norm $\left\Vert \cdot \right\Vert _{\infty }$). The third
assertion can be shown by observing that changes $F-F_{0}$ in the
conditional cdf of $Z_{i}$ are bounded by $C\left\Vert f-f_{0}\right\Vert
_{\infty }$ for some $C<\infty $. Since both $f_{0}$ and $f$ are bounded by
Assumption \ref{asszhat}(ii), the change $F^{-1}-F_{0}^{-1}$ is also bounded
by $C^{\prime }\left\Vert f-f_{0}\right\Vert _{\infty }$ for some $C^{\prime
}$ that is finite under Assumption \ref{asszhat}(ii). Thus the $Z_{i}$ are
Lipschitz in $f$.

For $R_{2}$, we observe that 
\begin{eqnarray*}
\left\vert R_{2}\right\vert  &\leq &H^{-1}\int \int \int \int \left\vert
f_{YXW}\left( y,x,w\right) \right\vert \left\vert \nabla \ln L\left(
y,x,z,w;\theta _{0},\omega \left( \theta _{0}\right) \right) \right\vert
dydxdzdw \\
&&\times n^{1/2}\max_{z,y,x,w}\left\vert \left( \left( \hat{f}\left( z|y,x,w;%
\hat{\kappa}\right) -\hat{f}\left( z|y,x,w;\kappa _{0}\right) \right)
-\left( f\left( z|y,x,w;\hat{\kappa}\right) -f\left( z|y,x,w;\kappa
_{0}\right) \right) \right) \right\vert 
\end{eqnarray*}%
where the integral is finite by Assumption \ref{asszhat} and \ref{asssupp}%
(ii), while the argument of the max can be written as: 
\begin{eqnarray*}
&&\left( \hat{f}\left( z|y,x,w;\hat{\kappa}\right) -\hat{f}\left(
z|y,x,w;\kappa _{0}\right) \right) -\left( f\left( z|y,x,w;\hat{\kappa}%
\right) -f\left( z|y,x,w;\kappa _{0}\right) \right)  \\
&=&\left( \nabla _{\kappa }^{\prime }\hat{f}\left( z|y,x,w;\dot{\kappa}%
\right) -\nabla _{\kappa }^{\prime }f\left( z|y,x,w;\dot{\kappa}\right)
\right) \left( \hat{\kappa}-\kappa _{0}\right) .
\end{eqnarray*}%
We have $\left\Vert \hat{\kappa}-\kappa _{0}\right\Vert =O_{p}\left(
n^{-1/2}\right) $ by Assumption \ref{assnuinf}\ while $\left\Vert \nabla
_{\kappa }^{\prime }\hat{f}\left( z|y,x,w;\dot{\kappa}\right) -\nabla
_{\kappa }^{\prime }f\left( z|y,x,w;\dot{\kappa}\right) \right\Vert
=o_{p}\left( 1\right) $ by Lemma \ref{lemwasass4}.

We can bound $R_{3}$ by expressing $\left( \nabla _{\kappa }^{\prime
}f\left( z|y,x,w;\dot{\kappa}\right) -\nabla _{\kappa }^{\prime }f\left(
z|y,x,w;\kappa _{0}\right) \right) $ using another mean value $\ddot{\kappa}$%
:%
\begin{eqnarray*}
\left\vert R_{3}\right\vert &\leq &n^{1/2}\left\Vert H\right\Vert ^{-1}\int \int \int \int \left\Vert
\nabla _{\kappa }\nabla _{\kappa }^{\prime }f\left( z|y,x,w;\ddot{\kappa}%
\right) \right\Vert \times\\
&& f_{YXW}\left( y,x,w\right) \left\Vert \nabla \ln L\left(
y,x,z,w;\theta _{0},\omega \left( \theta _{0}\right) \right) \right\Vert
dydxdzdw \, \left\Vert \hat{\kappa}-\kappa _{0}\right\Vert ^{2},
\end{eqnarray*}%
where $\left\Vert \hat{\kappa}-\kappa _{0}\right\Vert ^{2}=O_{p}\left(
n^{-1}\right) $ by Assumption \ref{assnuinf}, while the integral is bounded
by Assumption \ref{asssupp}(ii), \ref{asszhat}(i) \ref{asslink}(ii) and \ref%
{assdensmo}.
\end{proof}

\bigskip 

\begin{definition}
\label{defpowvec}For $t\in \mathbb{R}^{d}$ and $k\in \mathbb{N}^{d}$, let $%
t^{k}\equiv \prod_{i=1}^{d}\left( t_{i}\right) ^{k_{i}}$, $\left\vert
t\right\vert ^{k}\equiv \prod_{i=1}^{d}\left\vert t_{i}\right\vert ^{k_{i}}$%
, $\left\Vert k\right\Vert _{1}\equiv \sum_{i=1}^{d}\left\vert
k_{i}\right\vert $, $k!=\prod_{i=1}^{d}k_{i}!$ and $g^{\left( k\right)
}\left( t\right) =\frac{\partial ^{\left\Vert k\right\Vert _{1}}g\left(
t\right) }{\partial t_{1}^{k_{1}}\cdots \partial t_{d}^{k_{d}}}$.
\end{definition}

\begin{definition}
\label{defkorder}$K\left( \cdot \right) $ is a $d$-dimensional kernel of
order $r$ if $\int K\left( t\right) dt=1$, $\int K\left( t\right) t^{k}dt=0$
for $\sum_{i=1}^{d}k_{i}<r$\ and $\int \left\vert K\left( t\right)
\right\vert \left\vert t\right\vert ^{k}dt<\infty $ for $\sum_{i=1}^{d}k_{i}%
\leq r$.
\end{definition}

\begin{lemma}
\label{lemnewey}Let $K$ be a $d$-dimensional kernel of order $r$ and let $%
\hat{f}_{V}\left( v\right) =n^{-1}\sum_{i=1}^{n}h^{-d}K\left( \left(
v-V_{i}\right) /h\right) $. If the function $g\left( v\right) $ admits
uniformly continuous and bounded $r^{\text{th}}$ mixed derivatives, then,
\[n^{1/2}\left( \int \hat{f}_{V}\left( v\right) g\left( v\right) dv-E\left[
g\left( V\right) \right] \right) =n^{-1/2}\sum_{i=1}^{n}\left( g\left(
V_{i}\right) -E\left[ g\left( V\right) \right] \right) +O\left( h^{r}\right) .
\]
\end{lemma}

\bigskip

\begin{proof}[Proof of Lemma \protect\ref{lemnewey}]
Rewrite the left-hand side as 
\begin{equation*}
n^{1/2}\left( \int \hat{f}_{V}\left( v\right) g\left( v\right) dv-E\left[
g\left( V\right) \right] \right) =n^{1/2}\left( \frac{1}{n}%
\sum_{i=1}^{n}g\left( V_{i}\right) -E\left[ g\left( V\right) \right] +\frac{1%
}{n}\sum_{i=1}^{n}\left( g_{K}\left( V_{i}\right) -g\left( V_{i}\right)
\right) \right)
\end{equation*}%
where%
\begin{equation*}
g_{K}\left( v\right) \equiv \int \frac{1}{h^{d}}K\left( \frac{u-v}{h}\right)
g\left( u\right) du=\int K\left( t\right) g\left( v+th\right) dt
\end{equation*}%
by the change of variable $u=v+th$. Next, by a Taylor expansion,%
\begin{equation*}
\left\vert g_{K}\left( v\right) -g\left( v\right) \right\vert =\left\vert
\sum_{0\leq \left\Vert \ell \right\Vert _{1}<k}\frac{h^{\ell }}{\ell !}%
g^{\left( \ell \right) }\left( v\right) \int K\left( t\right) t^{\ell
}dt\right\vert +\left\vert \sum_{\left\Vert \ell \right\Vert
_{1}=k}h^{k}\int K\left( t\right) g^{\left( k\right) }\left( v+\tilde{t}%
h\right) \frac{t^{\ell }}{\ell !}dt\right\vert
\end{equation*}%
where the first term vanishes by the properties of a kernel of order $k$ and
second term is bounded by:%
\begin{equation*}
h^{k}\sum_{\left\Vert \ell \right\Vert _{1}=k}\frac{1}{\ell !}\left(
\sup_{u}\left\vert g^{\left( k\right) }\left( u\right) \right\vert \right)
\int \left\vert K\left( t\right) \right\vert \left\vert t\right\vert ^{\ell
}dt=O\left( h^{k}\right) \text{.}
\end{equation*}
\end{proof}

\bigskip

\begin{lemma}
\label{lemkernconv}Let $K$ be a $d$-dimensional kernel of order $r\geq 2$
that is Lipschitz and let $f_{V}\left( v\right) $ admits $r\geq 2$ uniformly
bounded continuous derivatives. Let $V\left( s\right) $ be stationary and
strongly mixing\footnote{%
This Lemma also holds under the weaker mixing conditions given %
\citet{carbon1997kernel}, but this extension is not spelled out here for
conciseness.} for $s\in \mathcal{S}$. Then,%
\begin{equation}
\sup_{v\in \mathcal{V}}\left\vert \hat{f}\left( v\right) -f\left( v\right)
\right\vert =O_{p}\left( \left( \frac{\ln n}{nh^{d}}\right) ^{1/2}\right)
+O\left( h^{r}\right)   \label{eqkernvarbias}
\end{equation}%
where $\mathcal{V}$ is a compact subset of $\mathbb{R}^{d}$ and $\hat{f}%
\left( v\right) =\left( nh^{d}\right) ^{-1}\sum_{i=1}^{n}K\left( \left(
V_{i}-v\right) /h\right) $ for $V_{i}\equiv V\left( S_{i}\right) $ where the
random $S_{i}$ take value in $\mathcal{S}$. Moreover, if $r>d$, selecting $%
h=n^{-1/\left( 2r\right) -\varepsilon }$ for $\varepsilon >0$ yields: $%
\sup_{v\in \mathcal{V}}\left\vert \hat{f}\left( v\right) -f\left( v\right)
\right\vert =o_{p}\left( n^{-1/4}\right) $ and $\sup_{v\in \mathcal{V}%
}\left\vert E\left[ \hat{f}\left( v\right) \right] -f\left( v\right)
\right\vert =o\left( n^{-1/2}\right) $.
\end{lemma}

\begin{proof}
The results in \citet{carbon1997kernel} do not direcly handle the case of
higher-order kernels. However, noting that$\left\vert \hat{f}\left( v\right)
-f\left( v\right) \right\vert \leq \left\vert \hat{f}\left( v\right) -E\left[
\hat{f}\left( v\right) \right] \right\vert +\left\vert E\left[ \hat{f}\left(
v\right) \right] -f\left( v\right) \right\vert $ we observe that Theorem 3.1
in \citet{carbon1997kernel} directly implies that the first term is $%
O_{p}\left( \left( \frac{\ln n}{nh^{d}}\right) ^{1/2}\right) $ under our
assumptions, regardless of the order of the kernel.\ Next, the bias term, $E%
\left[ \hat{f}\left( v\right) \right] -f\left( v\right) $, which does not
depend on the spatial correlation structure, can be calculated in the
standard way (e.g. \citet{andrews:semikern}) to yield $O\left( h^{r}\right) $
for a multivariate $r$ order kernel. The specific rates for $h=n^{-1/\left(
2r\right) -\varepsilon }$ can be shown by direct substitution.
\end{proof}

\bigskip

\begin{lemma}
\label{lemwasass4}Under Assumptions \ref{asssupp}, \ref{asskern2}, \ref%
{assdensmo} and \ref{asslink}, we have
\begin{eqnarray*}
\text{(i)}&& \sup_{y,x,w\in \mathcal{Y}\times 
\mathcal{X}\times \mathcal{W}}\left\vert \hat{f}_{Y,X|W}\left( y,x|w\right)
-f_{Y,X|W}\left( y,x|w\right) \right\vert =o_{p}\left( n^{-1/4}\right),\\
\text{(ii)}&& \sup_{y,x,z,w\in \mathcal{Y}\times \mathcal{X}\times \mathcal{Z}\times 
\mathcal{W}}\left\vert \hat{f}_{Y,X,Z|W}\left( y,x,z|w\right)
-f_{Y,X,Z|W}\left( y,x,z|w\right) \right\vert =o_{p}\left( n^{-1/4}\right),\\
\text{(iii)}&&\ \sup_{y,x,w\in \mathcal{Y}\times \mathcal{X}\times \mathcal{W}%
}\left\vert E\left[ \hat{f}_{Y,X|W}\left( y,x|w\right) \right]
-f_{Y,X|W}\left( y,x|w\right) \right\vert =o\left( n^{-1/2}\right) \\
\text{(iv)}&& \sup_{y,x,z,w\in \mathcal{Y}\times \mathcal{X}\times \mathcal{Z}\times 
\mathcal{W}}\left\vert E\left[ \hat{f}_{Y,X,Z|W}\left( y,x,z|w\right) \right]
-f_{Y,X,Z|W}\left( y,x,z|w\right) \right\vert =o\left( n^{-1/2}\right)\text{ and} \\
\text{(v)}&& \sup_{y,x,z,w,\kappa \in \mathcal{Y}\times \mathcal{X}\times \mathcal{Z}%
\times \mathcal{W}\times \mathcal{K}}\left\Vert \nabla _{\kappa }\hat{f}%
\left( z|y,x,w;\kappa \right) -\nabla _{\kappa }f\left( z|y,x,w;\kappa
\right) \right\Vert =o_{p}\left( 1\right) .
\end{eqnarray*}

\end{lemma}

\bigskip 

\begin{proof}
Consider the random vector
\[
V\equiv \left( \tilde{Y},\tilde{X},\tilde{Z}%
,\Delta S\right) =\left( G^{-1}\left( (Y,X,Y),W,\kappa _{0}\right)
,\Delta S\right)
\]
and observe that its joint density is $r$
times continuously differentiable and bounded by Assumptions \ref{assdensmo}.

Applying Lemma \ref{lemkernconv} to $V$ yields $\sup_{\tilde{y},\tilde{x},%
\tilde{z},\Delta s}\left\vert \hat{f}_{\tilde{Y},\tilde{X},\tilde{Z},\Delta
S}\left( \tilde{y},\tilde{x},\tilde{z},\Delta s\right) -f_{\tilde{Y},\tilde{X%
},\tilde{Z},\Delta S}\left( \tilde{y},\tilde{x},\tilde{z},\Delta s\right)
\right\vert =o_{p}\left( n^{-1/4}\right) $ and $\sup_{\tilde{y},\tilde{x},%
\tilde{z},\Delta s}\left\vert E\left[ \hat{f}_{\tilde{Y},\tilde{X},\tilde{Z}%
,\Delta S}\left( \tilde{y},\tilde{x},\tilde{z},\Delta s\right) \right] -f_{%
\tilde{Y},\tilde{X},\tilde{Z},\Delta S}\left( \tilde{y},\tilde{x},\tilde{z}%
,\Delta s\right) \right\vert =o\left( n^{-1/2}\right) $ (where the sups are
over compact sets) for the estimator (\ref{eqfhattilde}) in the main text
(for $\hat{\kappa}=\kappa _{0}$). By construction, these rates automatically
carry over to the estimator (\ref{eqfhatnotilde}), since the Jacobian
$=\left( \det \nabla _{(y,x,z)}^{\prime }G\left(
(y,x,z),w\right) \right) ^{-1}$
is bounded by Assumption %
\ref{asslink}. This establishes conclusions (i) and (ii) of the Lemma. A
similar reasoning can be used for $f_{Y,X|W}\left( y,x|w;\Delta
s\right) $ in order to establish (iii) and (iv).

For conclusion (v), we invoke Assumption \ref{asslink}(ii):
$\nabla _{\kappa } G^{-1}\left( (y,x,z),w,\kappa
\right) $
exists and is twice
uniformly continuously differentiable in $(y,x,z)$. Then, by the same reasoning as
above, we can use Lemma \ref{lemkernconv}, now with $r=2$, to conclude that $%
\left\Vert \nabla _{\kappa }\hat{f}\left( z,y,x,w;\kappa \right) -\nabla
_{\kappa }f\left( z,y,x,w;\kappa \right) \right\Vert $ is uniformly at most $%
o_{p}\left( 1\right) $, with uniformity in $\kappa $ holding because \ref%
{asslink}(ii) holds uniformly in $\kappa $.\ The same conclusion then holds
for $\nabla _{\kappa }\hat{f}\left( z|y,x,w;\kappa \right) -\nabla _{\kappa
}f\left( z|y,x,w;\kappa \right) $ since $f\left( z|y,x,w;\kappa \right)
=f\left( z,y,x|w;\kappa \right) /f\left( y,x|w;\kappa \right) $ with nonzero
denominator by Assumption \ref{asssupp}.
\end{proof}

\bigskip

The following Theorem collects the results found in Section 3 of %
\citet{bickel:boot}.

\begin{theorem}
\label{thbickel}Let $F$ denote an arbitrary cdf, $F_{0}$ denote the true
cdf, $F_{n}$ denote the empirical cdf for an iid sample of size $n$ and $%
F_{n}^{\ast }$ denote the empirical cdf of a bootstrap sample of size $n$
drawn wih replacement from a sample of size $n$. Let $g\left( F\right) $ be G%
\^{a}teau-differentiable at $F=F_{0}$ with derivative $\dot{g}\left(
F_{0}\right) $ representable as an integral:%
\begin{equation*}
\dot{g}\left( F_{0}\right) \left( F-F_{0}\right) \equiv \left[ \frac{%
\partial }{\partial \varepsilon }g\left( F_{0}+\varepsilon \left(
F-F_{0}\right) \right) \right] _{\varepsilon =0}=\int \psi \left(
x,F_{0}\right) dF\left( x\right) .
\end{equation*}%
If (i) $\int \left\Vert \psi \left( x,F_{0}\right) \right\Vert
^{2}dF_{0}\left( x\right) <\infty $ and (ii) $\int \left\Vert \psi \left(
x,F_{n}\right) -\psi \left( x,F_{0}\right) \right\Vert ^{2}dF_{n}\overset{as}%
{\longrightarrow }0$, then%
\begin{equation*}
n^{1/2}\left( g\left( F_{n}^{\ast }\right) -g\left( F_{n}\right) \right) 
\overset{d}{\longrightarrow }N\left( 0,\Omega \right)
\end{equation*}%
and%
\begin{equation*}
n^{1/2}\left( g\left( F_{n}\right) -g\left( F_{0}\right) \right)
=n^{-1/2}\sum_{i=1}^{n}\psi \left( x_{i},F_{0}\right) +o_{p}\left( 1\right) 
\overset{d}{\longrightarrow }N\left( 0,\Omega \right)
\end{equation*}%
for $\Omega =E\left[ \psi \left( x,F_{0}\right) \psi ^{\prime }\left(
x,F_{0}\right) \right] $.
\end{theorem}

This result can be extended to dependent data using a standard
\textquotedblleft blocking\textquotedblright\ device (e.g. %
\citet{carlstein:blockboot}, \citet{nordman2007optimal}).

\pagebreak

\section{Additional simulations}
\label{appaddsim}

\begin{table}[!ht]\caption{Simulation Results --- Effect of sieve truncation}
\begin{center}
\begin{tabular}{ | l | l | l | l |}
\hline
$\theta_1 = -3.5$	& {Mean} & {Standard deviation} & {RMSE} \\ \hline
	{Unweighted Spatial; low Sieve} & -4.38 & 0.24 & 0.91 \\ \hline
	{weighted Spatial; low Sieve} & -4.32 & 0.23 & 0.85 \\ \hline
	{Unweighted Spatial; medium Sieve} & -3.58 & 0.19 & 0.20 \\ \hline
	{weighted Spatial; medium Sieve} & -3.58 & 0.19 & 0.20\\ \hline
	{Unweighted Spatial; high Sieve} & -3.60 & 0.11 & 0.15 \\ \hline
	{weighted Spatial; high Sieve} & -3.59 & 0.08 & 0.12 \\ \hline
\end{tabular}

\medskip

\begin{tabular}{ | l | l | l | l |}
\hline
$\theta_2 = 2$ & {Mean} & {Standard deviation} & {RMSE} \\ \hline
	{Unweighted Spatial; low Sieve} & 2.26 & 0.06 & 0.27 \\ \hline
	{weighted Spatial; low Sieve} & 2.25 & 0.06 & 0.26  \\ \hline
	{Unweighted Spatial; medium Sieve} & 2.05 & 0.06 & 0.07 \\ \hline
	{weighted Spatial; medium Sieve} & 2.05 & 0.06 & 0.08\\ \hline
	{Unweighted Spatial; high Sieve} & 2.04 & 0.05 & 0.07 \\ \hline
	{weighted Spatial; high Sieve} & 2.04 & 0.04 & 0.06 \\ \hline
\end{tabular}

\medskip

\begin{tabular}{ | l | l | l | l |}
\hline
$\sigma_u = 1.3$ & {Mean} & {Standard deviation} & {RMSE} \\ \hline
	{Unweighted Spatial; low Sieve} & 1.17 & 0.06 & 0.14 \\ \hline
	{weighted Spatial; low Sieve} & 1.23 & 0.05 & 0.08 \\ \hline
	{Unweighted Spatial; medium Sieve} & 1.34 & 0.03 & 0.05 \\ \hline
	{weighted Spatial; medium Sieve} & 1.34 & 0.03 & 0.05\\ \hline
	{Unweighted Spatial; high Sieve} & 1.24 & 0.05 & 0.08 \\ \hline
	{weighted Spatial; high Sieve} & 1.23 & 0.04 & 0.07 \\ \hline
\end{tabular}
\end{center}
Unweighted Spatial: unweighted average spatial estimator; Weighted Spatial: optimally weighted estimator. With Section 3's notations, low Sieve means $i_n = j_n = 2$; medium Sieve means $i_n = j_n = 4$; high Sieve means $i_n = j_n = 6$
\end{table}

\pagebreak

\begin{table}[!ht]\caption{Simulations with median centering and probit. }
\label{mode_nonlin_results}
\[\begin{array}{ | l | l | l | l | l |}
\hline
\mbox{Specification}, \theta_1 & \mbox{Estimator} & \mbox{Mean} & \mbox{Standard deviation} & \mbox{RMSE} \\ \hline
	\multirow{4}{*}{Median centering, -3.5} & \mbox{Infeasible OLS} & -3.52 & 0.11 & 0.11\\ \cline{2-5}
    & \mbox{OLS} & -1.04 & 0.15 & 2.46 \\ \cline{2-5}
	& \mbox{Unweighted Spatial} & -3.72 & 0.19 & 0.30 \\ \cline{2-5}
	& \mbox{Weighted Spatial} & -3.75 & 0.30 & 0.39 \\ \hline
    \multirow{4}{*}{Probit, 0} & \mbox{Infeasible MLE} & -0.01 & 0.14 & 0.14 \\ \cline{2-5}
	& \mbox{MLE} & 0.44 & 0.12 & 0.46\\ \cline{2-5}
	& \mbox{Unweighted Spatial} & 0.10 & 0.04 & 0.11 \\ \cline{2-5}
	& \mbox{Weighted Spatial} & 0.09 & 0.04 & 0.10 \\ \hline
\end{array}
\]
\[\begin{array}{ | l | l | l | l | l |}
\hline
\mbox{Specification}, \theta_2 & \mbox{Estimator} & \mbox{Mean} & \mbox{Standard deviation} & \mbox{RMSE} \\ \hline
    \multirow{4}{*}{Median centering, 2} & \mbox{Infeasible OLS} & 2.01 & 0.03 & 0.03 \\ \cline{2-5}
    & \mbox{OLS} & 1.28	& 0.04 & 0.73 \\ \cline{2-5}
    & \mbox{Unweighted Spatial} & 2.04 & 0.06 & 0.07 \\ \cline{2-5}
	& \mbox{Weighted Spatial} & 2.05 & 0.08 & 0.09 \\ \hline
     \multirow{4}{*}{Probit, 1/3} & \mbox{Infeasible MLE} & 0.33 & 0.04 & 0.04\\ \cline{2-5}
	& \mbox{MLE} & 0.2 & 0.03 & 0.14\\ \cline{2-5}
	& \mbox{Unweighted Spatial} & 0.30 & 0.03 & 0.04 \\ \cline{2-5}
	& \mbox{Weighted Spatial} & 0.31 & 0.03	& 0.04 \\ \hline
\end{array}
\]
\[\begin{array}{ | l | l | l | l | l |}
\hline
\mbox{Specification}, \sigma_u & \mbox{Estimator} & \mbox{Mean} & \mbox{Standard deviation} & \mbox{RMSE} \\ \hline
	\multirow{4}{*}{Median centering, 1.3} & \mbox{Infeasible OLS} & 1.30 & 0.02 & 0.02\\ \cline{2-5}
    & \mbox{OLS} & 1.76 & 0.04 & 0.46\\ \cline{2-5}
    & \mbox{Unweighted Spatial} & 1.12 & 0.09 & 0.20 \\ \cline{2-5}
	& \mbox{Weighted Spatial} & 1.12 & 0.09 & 0.20 \\ \hline
\end{array}\]
The sample size is 1500; Unweighted Spatial refers to our unweighted average spatial estimator; Weighted Spatial is the optimally weighted average spatial estimator. 
\end{table}

\begin{table}[!ht]\caption{Simulations with semiparametric handling of a covariate.}\label{covariates_results}
\[\begin{array}{ | l | l | l | l |}
\hline
\theta_1 = -3.5 & \mbox{Mean} & \mbox{Standard deviation} & \mbox{RMSE} \\ \hline
	\mbox{Infeasible OLS} & -3.50 & 0.13 & 0.12 \\ \hline
	\mbox{OLS} & -0.72 & 0.14 & 2.71 \\ \hline
	\mbox{Unweighted Spatial} & -3.75 & 0.05 & 0.25 \\ \hline
	\mbox{Weighted Spatial}  & -3.75 & 0.04 & 0.24 \\ \hline
\end{array}\]
\[\begin{array}{ | l | l | l | l |}
\hline
\theta_2 = 2 & \mbox{Mean} & \mbox{Standard deviation} & \mbox{RMSE} \\ \hline
    \mbox{Infeasible OLS} & 2.00 & 0.04 & 0.04 \\ \hline
	\mbox{OLS} & 0.91 & 0.04 & 1.06 \\ \hline
	\mbox{Unweighted Spatial} & 2.09 & 0.02 & 0.09 \\ \hline
	\mbox{Weighted Spatial} & 2.09 &0.03 & 0.09 \\ \hline
\end{array}\]
\[\begin{array}{ | l | l | l | l |}
\hline
\sigma_u = 1.3& \mbox{Mean} & \mbox{Standard deviation} & \mbox{RMSE} \\ \hline
	\mbox{Infeasible OLS} & 1.30 & 0.02 & 0.02 \\ \hline
	\mbox{OLS} & 1.69 & 0.03 & 0.38 \\ \hline
	\mbox{Unweighted Spatial} & 1.35 & 0.06 & 0.07 \\ \hline
	\mbox{Weighted Spatial} & 1.35 & 0.06 & 0.08 \\ \hline
\end{array}\]
Sample size is 1500; Unweighted Spatial refers to our unweighted average spatial estimator; Weighted Spatial is the optimally weighted average spatial estimator.
\end{table}

\begin{table}[!ht]\caption{Coverage}\label{coverage_results}
\[\begin{array}{ | l | l | l | l |}
\hline
\mbox{Specification} & \mbox{Estimator} & \theta_1 & \theta_2 \\ \hline
\multirow{4}{*}{Median centering} & \mbox{Infeasible OLS} & 0.97 & 0.96\\ \cline{2-4}
    & \mbox{OLS} & 0 & 0\\ \cline{2-4}
    & \mbox{Unweighted Spatial} & 0.99 & 0.99 \\ \cline{2-4}
	& \mbox{Weighted Spatial} & 0.97 & 0.98 \\ \hline

\multirow{4}{*}{Probit} & \mbox{Infeasible MLE} & 0.96 & 0.95\\ \cline{2-4}
    & \mbox{MLE} & 0.03 & 0.03\\ \cline{2-4}
    & \mbox{Unweighted Spatial} & 0.69 & 0.72 \\ \cline{2-4}
	& \mbox{Weighted Spatial} & 0.68 & 0.74 \\ \hline

\multirow{4}{*}{Covariate} & \mbox{Infeasible OLS} & 0.96 & 0.98\\ \cline{2-4}
    & \mbox{OLS} & 0 & 0\\ \cline{2-4}
    & \mbox{Unweighted Spatial} & 0.52 & 0.89 \\ \cline{2-4}
	& \mbox{Weighted Spatial} & 0.50 & 0.90 \\ \hline
\end{array}\]
\caption*{Coverage performance of 95\% confidence intervals.}
\end{table}

\pagebreak 

\section{Block-bootstrap description}

We sample blocks of observations to preserve spatial correlations using the following algorithm.

\subsection{Algorithm}

For each bootstrap sample $b=1, \ldots, B$, \\
1. Sample a location among those in the sample, i.e., uniformly sample from $\{s_i, i=1, \ldots, n\}$ and center a block of size $l_1 \times l_2$ around it. Record the block structure and the observations within it. Repeat the step until the resulting collection of blocks contains $n$ observations in total. \\
2. Arrange the blocks in a disposition similar to the original spatial structure, i.e., (i) start laying out blocks horizontally until their combined length reaches the horizontal length of the (first $l_2$ units of the) original spatial structure; (ii) repeat with a second row of blocks, etc., until blocks are exhausted. \\
3. Compute the estimator $\hat{\theta}_b$ using the bootstrapped sample constructed from step 1-2. \medskip

The procedure provides a collection of estimates $\{\hat{\theta}_1, \ldots, \hat{\theta}_{B}\}$. One can then compute bootstrap standard errors and use these to form a 95\% confidence interval.

%$\theta = (-3.50, 0.20, 0.20,-0.05)$

\end{document}